\documentclass[aps,twocolumn,notitlepage,pra,superscriptaddress,footinbib]{revtex4-1}

\usepackage[usenames,dvipsnames]{color}

\usepackage{gensymb} 
\usepackage{physics} 
\usepackage{amsmath} 
\usepackage{amssymb} 
\usepackage{bm} 
\usepackage{bbm} 
\usepackage{braket} 
\usepackage[version=4]{mhchem} 
\usepackage{latexsym} 
\usepackage{tensor}
\usepackage{tcolorbox}
\usepackage{mathtools}
\usepackage{mathrsfs}

\usepackage{amsthm,thmtools}
\declaretheorem[
]{theorem}

\usepackage{CJKutf8}

\newcommand{\bh}[1]{\hat{\pmb{#1}}} 
\newcommand{\bs}[1]{\pmb{#1}} 

\usepackage{graphicx} 
\graphicspath{{./figures/}} 
\usepackage{stackengine} 
\usepackage[caption=false]{subfig} 
\usepackage[inline]{enumitem} 
\setlist[enumerate,1]{label={(\roman*)}} 
\setlist{nolistsep} 
\usepackage{comment} 

\usepackage{multirow}
\usepackage{hhline}

\usepackage{hyperref}
\usepackage{natbib}
\hypersetup{
 breaklinks=true,
 colorlinks=true,
 linkcolor=blue,
 filecolor=magenta,
 urlcolor=cyan,
}
\usepackage[capitalize]{cleveref}
\crefalias{section}{appendix}

\begin{document}
\begin{CJK}{UTF8}{gbsn}
\renewcommand{\labelenumii}{\theenumii}
\renewcommand{\theenumii}{\theenumi.\arabic{enumii}.}

\newcommand{\QuICS}{Joint Center for Quantum Information and Computer Science, National Institute of Standards and Technology and
 University of Maryland, College Park, Maryland 20742, USA}
\newcommand{\JQI}{Joint Quantum Institute, National Institute of Standards and Technology and
 University of Maryland, College Park, Maryland 20742, USA}
\newcommand{\USNA}{Department of Physics, United States Naval Academy, Annapolis, MD 21402, USA}
\newcommand{\harvard}{Department of Physics, Harvard University, Cambridge, MA 02138, USA}

\newcommand{\thetitle}{Optimally learning functions in interacting quantum sensor networks}

\title{\thetitle}

\author{Erfan~Abbasgholinejad}\thanks{These two authors contributed equally.}
\affiliation{\QuICS}
\affiliation{\JQI}
\author{Sean~R.~Muleady}\thanks{These two authors contributed equally.}
\affiliation{\QuICS}
\affiliation{\JQI}
\author{Jacob~Bringewatt}
\affiliation{\USNA}
\affiliation{\harvard}
\author{Anthony~J.~Brady}
\affiliation{\QuICS}
\affiliation{\JQI}
\author{Yu-Xin Wang (王语馨)}
\affiliation{\QuICS}
\author{Ali~Fahimniya}
\affiliation{\QuICS}
\affiliation{\JQI}
\author{Alexey~V.~Gorshkov}
\affiliation{\QuICS}
\affiliation{\JQI}

\date{\today}

\begin{abstract}

Estimating extensive combinations of local parameters in distributed quantum systems is a central problem in quantum sensing, with applications ranging from magnetometry to timekeeping. While optimal strategies are known for sensing non-interacting Hamiltonians in quantum sensor networks, fundamental limits in the presence of uncontrolled interactions remain unclear. Here, we establish optimal bounds and protocols for estimating a linear combination of local parameters of Hamiltonians with arbitrary, unknown interactions. In the process, we more generally establish bounds for learning any linear combination of Hamiltonian coefficients for arbitrary, commuting terms. Our results unify and extend existing bounds for non-interacting qubits and multimode interferometers, providing a general framework for distributed sensing and Hamiltonian learning in realistic many-body systems.

\end{abstract}

\maketitle
\end{CJK}

Precise characterization of the dynamics of a physical system is a fundamental task in the natural sciences. For a quantum system, accurate determination of the underlying Hamiltonian underpins diverse tasks including quantum sensing and metrology \cite{degen_quantum_2017, pirandola_advances_2018}, characterization and control of many-body systems, and the reliable operation of quantum computing devices \cite{preskill_quantum_2018}. In many relevant scenarios, desired information regarding the system is distributed across a network of quantum sensors, such as for ensembles of spatially proximate qubits in atomic, molecular, optical and solid-state platforms. While these systems often couple to local, external fields in known ways, they can also experience many-body interactions that are difficult to isolate or characterize directly; realizations include color center ensembles \cite{bar2013solid,Choi2020:DROIDtheory,Zhou2020:IntSpins, Gao2025:AmpEcho, Wu2025:SpinSqueezing}, neutral atom and polar molecule arrays interacting via dipolar or collisional interactions~\cite{chomaz_dipolar_2023,carroll_observation_2025}, and superconducting or Rydberg-based processors exhibiting crosstalk \cite{bluvstein_logical_2024,saxena2024boundary}. As such, efficiently extracting useful information in realistic many-body settings remains a key challenge.

A general task in distributed quantum systems is the estimation of extensive observables, e.g. linear combinations of local fields~\cite{eldredge_optimal_2018,proctor2017networked,proctor_multiparameter_2018,suzuki2020quantum,ehrenberg_minimum_2022,bringewatt_optimal_2024}, which arise naturally in gradient estimation~\cite{bate_experimental_2025,bothwell_resolving_2022}, clock network synchronization~\cite{komar_quantum_2014}, or spatial field inference~\cite{eldredge_optimal_2018}. More generally, many sensor network problems also reduce to this task, including estimating one or more analytic functions~\cite{qian2019heisenberg,bringewatt_protocols_2021} or learning linear combinations of correlated parameters~\cite{qian_optimal_2021,hamann2022approximate}. While optimal strategies and bounds are well established for sensing \emph{non-interacting} Hamiltonians~\cite{eldredge_optimal_2018,ehrenberg_minimum_2022}, these remain unknown for Hamiltonians with uncontrolled interactions. Prior studies examine bounds and strategies for encoding a \emph{single} parameter in an interacting Hamiltonian~\cite{boixo2007generalized,Zhou2020:IntSpins,Puig2025:2bodyIntSens}, but no framework exists for the broader multi-parameter problem.

Here, we solve this problem in a general setting where each of $n$ qubit sensors couples to an unknown local parameter through a known generator, while also experiencing arbitrary, unknown, time-independent interactions. The goal is to estimate a known linear function of the local parameters using quantum dynamics together with arbitrary control, including the use of known interactions, ancilla qubits, and adaptive protocols---a problem we refer to as ``Interacting Sensing'' (\cref{fig:fig1}~(a)). By leveraging Hamiltonian reshaping techniques to remove all non-commuting interaction terms,~\cite{chen_concentration_2021, campbell_random_2019, ma_learning_2024}, we reduce this to another problem we term ``Diagonal Learning''. This consists of estimating a linear function of unknown parameters coupled to a probe state via arbitrary commuting Hamiltonian terms. We derive optimal bounds, and offer a corresponding optimal strategy, which does not require full tomography of the system or learning of the individual Hamiltonian terms. Beyond offering a solution to the Interacting Sensor problem, our setup for the Diagonal Learning problem generally subsumes prior setups for function estimation in distributed quantum systems, including for networks of non-interacting qubits~\cite{eldredge_optimal_2018,ehrenberg_minimum_2022} or Mach-Zehnder interferometers~\cite{proctor_multiparameter_2018,ge_distributed_2018,Xia2020:RadioQSN,Guo2020:DQSphase,Liu2021:DQSphase,Kim2024:dqsPhase,bringewatt_optimal_2024}. Our work broadly establishes a framework for understanding of the limits of sensing extensive quantities in multi-parameter, many-body quantum systems.

\begin{figure}
\includegraphics[width=0.49\textwidth]{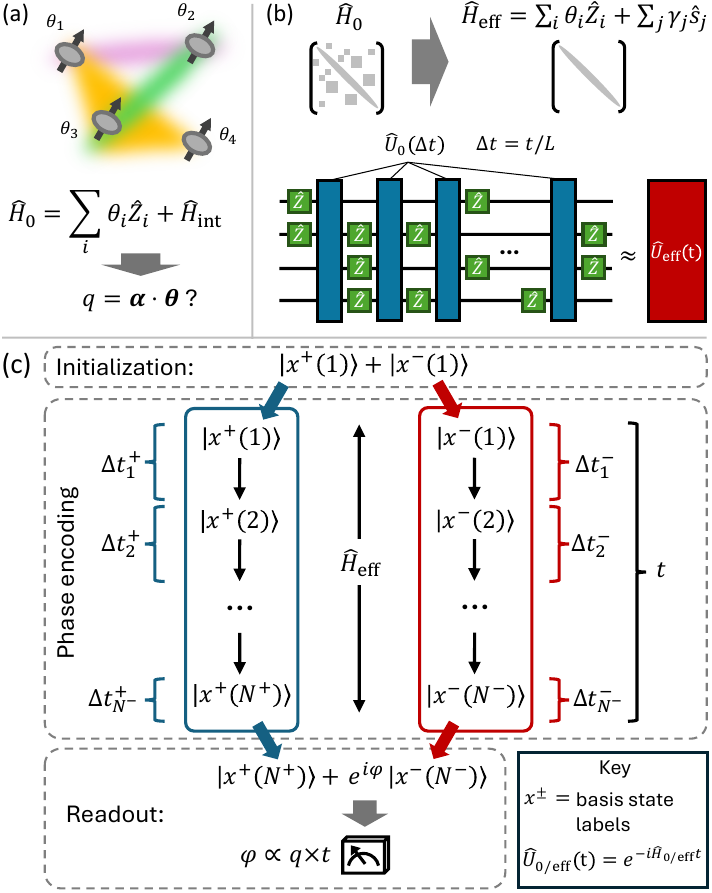}
\caption{(a) Schematic of an interacting sensor network, where each of $n$ qubits couples to a unique phase $\theta_i$, with \emph{unknown}, arbitrary multi-body interactions between them, $\hat{H}_{\rm int}$. The task of the Interacting Sensing problem is to learn a linear function $q = \bs{\alpha}\cdot\bs{\theta}$ of the phases for some real $\bs{\alpha}$. (b) To convert the Interacting Sensing problem to the Diagonal Learning problem, we first utilize randomized Trotter evolution under the non-diagonal Hamiltonian $\hat{H}_0$, conjugated by random Pauli-$Z$ strings. This engineers an effective, diagonal Hamiltonian $\hat{H}_{\rm eff}$, consisting of a sum over $\rm{poly}(n)$ Pauli-$Z$ strings $\hat{s}_j\in \{\hat{I},\hat{Z}\}^{\otimes n}$, so we have effective evolution $\hat{U}_{\rm eff}(t) = \exp\{-i\hat{H}_{\rm eff}t\}$. (c) Sketch of our optimal GHZ protocol for the Diagonal Learning problem, consisting of evolution via $\hat{U}_{\rm eff}(t)$, where $\hat{H}_{\rm eff}$ more generally may include any number of Pauli-$Z$ strings, interspersed with coherent switches between computational basis states for either the left (blue) or right (red) half of our initial superposition. The sequence of computational basis states $\ket{x^\pm(\mu)}$ and corresponding time intervals $\Delta t_\mu^\pm$ for the left/right halves of the superposition are selected to engineer a maximal total phase difference $\varphi \propto q \times t$, which can then be measured.}
\label{fig:fig1}
\end{figure}

\textit{Interacting Sensing and Diagonal Learning.---}For both Interacting Sensing and Diagonal Learning, the task is to measure a linear function $q=\bs{\alpha}\cdot\bs{\theta}$, where $\bs{\alpha}\in\mathbb{R}^m$ is a known vector of coefficients and $\bs{\theta}\in\mathbb{R}^m$ is a vector of unknown parameters encoded into a quantum state $\hat{\rho}_{\bs{\theta}}$.
The encoding process and number of unknown parameters differ between the two settings. For Interacting Sensing, we consider $n$ qubit sensors, each coupled to their own unknown parameter $\theta_i$ for $i = 1,..,n$ (so $m=n$). These parameters are encoded into a probe state $\hat{\rho}_0$ via unitary evolution for a time $t$ via the Hamiltonian
\begin{align}\label{eq:H_interacting_sensing}
\hat{H}(s) &= \hat{H}_0 + \hat{H}_{c}(s), \quad \hat{H}_0 = \sum_{i=1}^n \theta_i \hat{Z}_i + \hat{H}_{\text{int}},
\end{align}
where $\hat{Z}_i$ is the Pauli-$Z$ operator for the $i$-th qubit\footnote{Each $\theta_i$ may couple to any arbitrary axis on the Bloch sphere; however, we can always map this to $Z$ without loss of generality by appropriately redefining our local basis.}, $\hat{H}_{c}(s)$ for $s\in[0,t]$ is a freely-chosen $\bs{\theta}$-independent control Hamiltonian, potentially coupled to arbitrary number of ancilla qubits. $\hat{H}_{\text{int}}$ is an unknown, time- and $\bs{\theta}$-independent interaction Hamiltonian\footnote{In fact, $\hat{H}_{\text{int}}$ need not be $\bs{\theta}$-independent for our protocol to work, though its optimality is no longer guaranteed.} of the form
\begin{equation}
 \hat{H}_{\text{int}} = \sum_{j} \gamma_j \hat{P}_j, 
\end{equation}
where $\hat{P}_j$ are some arbitrary set of $O (\text{poly}(n))$ Pauli strings on $n$ qubits, $\hat{P}_j \notin \{\hat{Z}_i\}$. We assume the generators of the interaction Hamiltonian are known based on the sensor architecture, but the coupling strengths $\gamma_j$ are unknown. For example, ${\hat{P}_j}$ could include all weight-2 and weight-3 Pauli strings, corresponding to all possible two-body and three-body interactions.

In the Diagonal Learning problem, there are $m \leq 2^n-1$ unknown parameters $\bs{\theta}\in\mathbb{R}^m$ coupled to a probe state $\hat{\rho}_0$ for a time $t$ via the Hamiltonian
\begin{equation}\label{eq:H_stabilizer_sensing}
\hat{H}(s) = \hat{H}_\mathcal{S} + \hat{H}_c(s), \quad \hat{H}_\mathcal{S} = \sum_{j=1}^m\theta_j \hat{s}_j,
\end{equation}
where the generators $\hat{s}_j$ are any \emph{non-identity} elements in the stabilizer group of Pauli-$Z$ strings $\mathcal{S} = \{ \hat{I},\hat{Z}\}^{\otimes n}$ on $n$ qubits, 
and $\hat{H}_c(s)$ is again an arbitrary $\bs{\theta}$-independent control Hamiltonian.

\textit{Interacting Sensing Reduces to Diagonal Learning.---}Interacting Sensing can be reduced to Diagonal Learning by applying a Hamiltonian reshaping technique that transforms a general Hamiltonian into an effective Hamiltonian $\hat{H}_{\text{eff}}$ whose terms all commute \cite{ma_learning_2024}. This reshaping is particularly useful for quantum sensing and learning tasks, where commuting or diagonal Hamiltonians allow for simpler estimation strategies and more accessible information extraction. 
The key idea is to randomly apply Pauli-$Z$ operations between short time steps of evolution under $\hat{H}_0$, effectively projecting the dynamics onto a subspace generated by commuting operators (\cref{fig:fig1}~(b)).
This reshaping strategy has a natural connection to Hamiltonian simulation, where the goal is typically to approximate the time evolution under a complex Hamiltonian using simpler operations. In our case, we are interested in the reverse direction: transforming a complex, interacting Hamiltonian into one with a simpler structure that is easier to learn from. 

In particular, observe that, by randomly conjugating a Hamiltonian $\hat{H}$ with uniformly sampled elements from the stabilizer group $\mathcal{S}$, and averaging the result, we obtain an effective Hamiltonian that retains only those terms in $\hat{H}$ that commute with all elements of $\mathcal{S}$. Formally (see also Lemma 1 in Ref.~\cite{ma_learning_2024}):
\begin{equation}\label{eq:proj_lemma}
 \frac{1}{N} \sum_{\hat{s} \in \mathcal{S}} \hat{s} \hat{P} \hat{s} =
 \begin{cases}
 \hat{P}, & \text{if } \hat{P} \in \mathcal{S}, \\
 0, & \text{otherwise},
 \end{cases}
\end{equation}
where $N := 2^n$ is the cardinality of $\mathcal{S}$, and $\hat{P}$ is an arbitrary Pauli string.

Applying \cref{eq:proj_lemma} to $\hat{H}_0$ from \cref{eq:H_interacting_sensing}, the effective reshaped Hamiltonian becomes
\begin{equation}\label{eq:Heff}
 \hat{H}_{\text{eff}} = \frac{1}{N} \sum_{\hat{s} \in \mathcal{S}} \hat{s} \hat{H}_0 \hat{s} = \sum_{i=1}^n \theta_i \hat{Z}_i + \sum_{j\,|\,\hat{P}_j \in \mathcal{S}/\{\hat{Z}_i\}} \gamma_j \hat{P}_j,
\end{equation}
for some unknown coefficients $\gamma_j$, which contains only mutually commuting terms and is diagonal in the computational basis. Hence, this procedure can also be understood as a method that removes all off-diagonal (non-commuting) elements from $\hat{H}_0$ while preserving the generators that encode the parameters of interest. Crucially, observe that \cref{eq:Heff} is of the same form as $\hat{H}_\mathcal{S}$ in \cref{eq:H_stabilizer_sensing}, indicating that, if we can perform the Hamiltonian reshaping, we have reduced Interacting Sensing to Diagonal Learning.

Each term in $\hat{H}_{\text{eff}}$ is of the form $\hat{s} \hat{H}_0 \hat{s}$ for some $\hat{s} \in \mathcal{S}$. Since, in the Interacting Sensing problem, we have access to time evolution under $\hat{H}_0$, we can effectively simulate evolution under each $\hat{s} \hat{H}_0 \hat{s}$ by sandwiching the evolution with $\hat{s}$. This allows us to implement a Trotterized or QDRIFT-like protocol \cite{berry_time-dependent_2020, ma_learning_2024, chen_concentration_2021, campbell_random_2019} where we randomly sample $\hat{s} \in \mathcal{S}$ and perform the evolution
\begin{align}
\hat{V} &= \prod_{k=1}^L \hat{s}_ke^{-i \hat{H}_0 \Delta t}\hat{s}_k
= \prod_{k=1}^L e^{-i \hat{s}_k \hat{H}_0 \hat{s}_k \Delta t} ,
\end{align}
where $\hat{s}_k$ are sampled uniformly from $\mathcal{S}$ and $\Delta t = t / L$. Throughout this manuscript, we use a convention where $\prod_{k=1}^L \hat O_k = \hat O_L \dots \hat O_1$. The applications of stabilizers $\hat{s}_k$ between time steps can be implemented via the control Hamiltonian $\hat{H}_c(s)$ in \cref{eq:H_interacting_sensing}, which we let be zero, except for negligibly short periods of time when it is used to apply gates. Although $\mathcal{S}$ contains $2^n$ elements, it has been shown in prior work \cite{chen_concentration_2021} that, for approximating the target unitary
\begin{equation}
 \hat{U} = e^{-i \hat{H}_{\text{eff}} t}
\end{equation}
to additive error $\epsilon$ (i.e. $\|\hat{U}-{\hat{V}}\| \leq \epsilon$ for spectral norm $\|\cdot\|$) 
with probability $1-\delta$, the number of Trotter steps should satisfy
\begin{equation}
 L \gtrsim C\left((n\lambda^2t^2/\epsilon^2)\log(1/\delta)\right),
\end{equation}
for some constant $C$, where $\lesssim$ denotes an asymptotic bound for large $L$ and where $\lambda = \|\hat{H}_0\| = O(\mathrm{poly}(n))$. Thus, only $O(\mathrm{poly}(n))$ Trotter steps suffice to approximate the full evolution under $\hat{H}_{\text{eff}}$ to arbitrarily high precision. Therefore, this method enables efficient reduction of Interacting Sensing to Diagonal Learning up to an arbitrarily small error $\epsilon$ with high probability.

We briefly note that one possible alternative to the above method is to instead continuously apply strong $\hat{Z}$ fields to each qubit, where the corresponding field strengths are incommensurate with each other. For a suitable choice of sufficiently strong fields, this ensures that in the combined local interaction frame the applied fields, any Pauli string in $\hat{H}_{\rm int}$ containing at least one $\hat{X}$ or $\hat{Y}$ operator will time-average to $0$, to first order in the rotating wave approximation.

Next, we turn to optimal protocols for Diagonal Learning. Once these protocols are established, we return to the role of the error $\epsilon$ and how it affects the overall precision of an estimate of $q$ in the Interacting Sensing setting.

\textit{Optimal Diagonal Learning.---}Here, we present a class of optimal protocols that provide an unbiased estimate for $q=\bs{\alpha}\cdot\bs{\theta}$ in the Diagonal Learning problem. These protocols are optimal in the sense that they leverage probe states that saturate the quantum Cram\'er-Rao bound~\cite{helstrom1976quantum,Paris2009:QuMetrology,holevo2011probabilistic}, which provides the fundamental lower bound on the variance of an unbiased estimator $q_{\rm est}$ of $q$. We leave the detailed derivation of the bound to \cref{app:opt_protocol}, and simply introduce the result. 

First, we define some preliminary notation: let $\bh{s}'$ denote \emph{all} $N = 2^n$ stabilizers in $\mathcal{S}$, indexed in such a way that $\hat{s}_j' = \hat{s}_j$ for $1 \leq j \leq m$, and $\hat{s}'_0 = \hat{I}^{\otimes n}$. Also, let $\ket{x}$ for $x \in \{0,1\}^{\otimes n}$ denote computational basis states in the $N$-dimensional Hilbert space of the Diagonal Learning problem. We then have the $N \times N$ (i.e.~$2^n \times 2^n$) matrix $h'$ of stabilizer eigenvalues such that $h'_{j,x} := \braket{x|\hat{s}_j'|x}$, where we label the column index in binary using the bit string $x$. Up to a permutation of the rows and columns, this is exactly the $2^n \times 2^n$-dimensional Hadamard transformation matrix
\begin{align}\label{eq:hadamard}
 H = \begin{pmatrix}
 1 & 1 \\ 1 & -1
 \end{pmatrix}^{\otimes n},
\end{align}
which is equivalent to an $n$-dimensional discrete Fourier transform over $\{0,1\}^{\otimes n}$.

Now, let $\hat{q}_{\rm est}$ be an unbiased estimator for $q$, i.e. for $\delta \hat{q} := \hat{q}_{\rm est} - q$, then $\braket{\delta \hat{q}} = 0$ where $\braket{\cdot}$ denotes the expectation over measurement outcomes. We can now state our result:
\begin{align}\label{eq:Bnd}
 \braket{\delta \hat{q}^2} \geq \min_{\bs{a}\in \mathcal{A}}\frac{\norm{\bs{a}}_1^2}{4t^2}
\end{align}
where $\norm{\bs{a}}_p = (\sum_x |a_x|^p)^{(1/p)}$ is the $\ell_p$-norm. Here, $\mathcal{A}$ is the set of Hadamard transformed vectors $\bs{a} := (h')^T\bs{\alpha}'/N$, where $\bs{\alpha}':=(0,\bs{\alpha},\tilde{\bs{\alpha}})\in\mathbb{R}^N$ for $\tilde{\bs{\alpha}}$ an arbitrary vector $\in\mathbb{R}^{N-m-1}$, and the minimization is performed over all choices for $\tilde{\bs{\alpha}}$. We note that the inverse transformation is given by $\bs{\alpha}' = h' \bs{a}$.

The crucial dependence of the optimal precision on the Hadamard transformation of the coefficient vector of the desired function $\bs{\alpha}$ implies close connections to the classical theory of compressed sensing \cite{foucart2013invitation,donoho2006compressed}. In general, the minimization problem of $\norm{\bs{a}}_1$ over all choices for $\tilde{\bs{\alpha}}$ amounts to a basis pursuit problem \cite{chen2001atomic}, which can be solved via various linear programming techniques. For the limiting case where $m=N-1$ (i.e.~our generators consist of \emph{all} non-trivial Pauli-Z strings), $\bs{\alpha}'$ is uniquely determined by $\bs{\alpha}$, and so $\mathcal{A}$ consists of only a single element, and thus no optimization is required.

We now offer a protocol whose sensitivity saturates the bound in Eq.~\eqref{eq:Bnd}. Given an optimal choice of $\bs{a}\in \mathcal{A}$ (with minimal $\ell_1$-norm), form two disjoint sets of bit strings: $\Lambda^+ := \{x\,|\,a_x > 0\}$ and $\Lambda^- := \{x\,|\,a_x< 0\}$. The protocol consists of coherently switching during the encoding process between GHZ-like superposition states of the form $\ket{\psi_{x,y}(\theta)} \equiv (\ket{x} + e^{i\theta}\ket{y})/\sqrt{2}$ for $x\in\Lambda^+$ and $y\in\Lambda^-$. The switching operations are performed at appropriate times---specified below---in order to engineer a relative phase difference between the two halves of the superposition proportional to $q$. The required unitary switching operations $\hat{S}_{x,x'}$ ($\hat{S}_{y,y'}$) take $\ket{x}$ to $\ket{x'}$ ($\ket{y}$ to $\ket{y'}$) when applied to the state $\ket{\psi_{x,y}(\theta)}$, without affecting the other half of the superposition; for example, the unitaries
\begin{align}
 \hat{S}_{x,x'} = \ket{x}\bra{x'} + \ket{x'} \bra{x} + \sum_{x''\neq x,x'} \ket{x''}\bra{x''}
\end{align}
exchange basis states labeled by $x$, $x'$ while acting trivially on all other states.
For the state $\ket{\psi_{x,y}(\theta)}$ with $x' \neq y$ ($y' \neq x$), $\hat{S}_{x,x'}$ ($\hat{S}_{y,y'}$) can be effectively executed by a sequence of CNOT or CNOT$_0$---CNOT with control on $\ket{0}$---gates on qubits where $x$ and $x'$ ($y$ and $y'$) differ. These gates may be conditioned on either 1) any qubit whose current value differs between the two halves of the superposition or 2) an ancilla register, where the relevant joint state is $\ket{x}\otimes\ket{0} + e^{i\theta}\ket{y}\otimes\ket{1}$.
Note that if $x,x'\in\Lambda^+$ ($y,y'\in\Lambda^-$), $\hat{S}_{x,x'}$ ($\hat{S}_{y,y'}$) can change only the left (right) branch of the superposition in $\ket{\psi_{x,y}(\theta)}$.

Now, let $x^\pm(\mu) \in \Lambda^\pm$ be enumerations of the states in $\Lambda^\pm$, where the index $\mu$ runs over $1 \leq \mu \leq N^\pm$ for $\Lambda^\pm$, with $N^\pm := |\Lambda^\pm|$. Our optimal protocol is as follows: first, prepare a probe state $\ket{\psi_{x^+(1),x^-(1)}(0)}$. Then, while coupled to the Hamiltonian in \cref{eq:H_stabilizer_sensing} for total time $t$, use the control Hamiltonian to apply the operations $\hat{S}_{x^+(\mu),x^+(\mu+1)}$ in sequence after subsequent durations of length $\Delta t_\mu^+=2t(a_{x^+(\mu)}/{\norm{\bs{a}}_1})$, for $\mu = 1$ to $N^+-1$. Meanwhile, we also apply the operations $\hat{S}_{x^-(\nu),x^-(\nu+1)}$ in sequence after subsequent durations of length $\Delta t_\nu^-=2t(|a_{x^-(\nu)}|/{\norm{\bs{a}}_1})$, for $\nu = 1$ to $N^--1$.

As $h'_{0,x}=1$ for all $x$, and $\alpha'_0=0$ (by assumption), it follows that $\sum_x a_x = 0$. Thus, $
\sum_{x\in \Lambda^+} a_x = \sum_{x\in\Lambda^-} |a_x| =\norm{\bs{a}}_1/2$, and so $\sum_\mu \Delta t_\mu^\pm = t$. In other words, our sequence of switching operations and corresponding evolution duration for each GHZ-like superposition are chosen so that the requisite evolution time of each half of the superposition is exactly $t$.

The overall result of these operations is the state $\ket{\psi_{x^+(N^+),x^-(N^-)}(\varphi)}$ (up to an irrelevant global phase), where, due to the timing of the switching operations, the relative phase $\varphi$ is proportional to the function $q$ of interest, as desired:
\begin{align}
\varphi
&=\frac{2t}{\norm{\bs{a}}_1}\sum_{j=1}^m\theta_j \left(h' \bs{a}\right)_j = \frac{2t}{\norm{\bs{a}}_1}q.
\end{align}

To measure this relative phase, we note that qubits whose state is identical between the left and right halves of the superposition are disentangled from the rest of the system. The relative phase on the remaining qubits may be read out, e.g. via a standard parity measurement~\cite{bollinger_optimal_1996,sackett_experimental_2000,eldredge_optimal_2018}\footnote{We note that in place of a multi-qubit parity measurement, one can instead utilize CNOT and CNOT$_0$ gates to disentangle the superposition state, transferring the relative phase to a single qubit that can be read out.}.
In the limit of small signal, $qt \rightarrow 0$, the resulting measurement sensitivity exactly saturates Eq.~\eqref{eq:Bnd}, so our protocol is optimal. We provide a schematic of the protocol in Fig.~\ref{fig:fig1}~(c). 

To account for the phase wrapping that occurs for large evolution times, one can use the standard technique of robust phase estimation~\cite{higgins_demonstrating_2009,berry_how_2009,Kimmel2015:RobustPhase,belliardo2020achieving} to obtain an estimate of $q$ that saturates the bound in \cref{eq:Bnd}, up to relatively small constants.

\textit{The universality of Diagonal Learning.---}Diagonal Learning subsumes all previous work on function estimation in quantum sensor networks with commuting generators, including non-interacting qubit sensors with local generators~\cite{eldredge_optimal_2018} and networks of Mach Zehnder interferometers~\cite{proctor_multiparameter_2018,bringewatt_optimal_2024}.
In \cref{app:opt_protocol}, we show that for Hamiltonians of the form
\begin{equation}
\hat{H}(s) = \sum_{j=1}^m \theta_j \hat{g}_j + \hat{H}_{c}(s),
\end{equation}
with commuting generators $\hat{g}_j$ in dimension $N$, the bound for estimating $q=\bs \alpha \cdot \bs \theta$ in Eq.~(\ref{eq:Bnd}) still holds, with the modification that the minimization of $\|\bs a\|_1$ is over $\mathcal{A} = \{\bs a \mid h_g \bs a = \bs \alpha, \bs{1}^T\bs{a} = 0\}$. Here, $h_g \in \mathbb{R}^{m \times N}$ is the matrix of eigenvalues of the generators in the basis $\{\ket{k}\}_{k=1}^N$ where the generators are diagonal, i.e., $(h_{g})_{j,k} = \bra{k}\hat{g}_j\ket{k}$. This minimization remains a linear program and generally can be solved in $O(\mathrm{poly}(m)N)$ time \cite{boyd2004convex}. As in the stabilizer setting, there exists an optimal protocol that achieves the bound for $q$, and the optimal protocol is determined by the optimal $\bs a$. Moreover, as shown in \cref{app:sparsity}, the optimal solution satisfies $\|\bs a\|_0 \leq m+1$. Thus for $m=\mathrm{poly}(n)$, the optimal $\bs{a}$ is relatively sparse, and we only require $\mathrm{poly}(n)$ coherent switching operations for the corresponding optimal protocol. While efficiently determining this optimal $\bs a$ in general remains open for exponentially large $N$ in many cases this optimization is analytically solvable. Furthermore, the search can be restricted to a subset of $\mathcal{A}$, within which we can still obtain a valid, but possibly suboptimal, protocol with known performance.

For instance, in \cref{app:universality}, we show that our bound in \cref{eq:Bnd} reproduces the known results for non-interacting qubit sensors~\cite{eldredge_optimal_2018} and for networks of Mach-Zehnder interferometers~\cite{proctor_multiparameter_2018,bringewatt_optimal_2024}. 
Moreover, in \cref{app:bos_reshaping}, we demonstrate that Hamiltonian reshaping techniques can be used to eliminate unwanted non-commuting interactions in bosonic systems. Our results not only provide a unified bound and protocol for previous studies of sensor networks with commuting generators but also extend them to networks of sensors in the presence of unknown, arbitrary interactions.

\textit{Interacting Sensing protocol.---}The reduction from Interacting Sensing to Diagonal Learning via Hamiltonian reshaping implies that we now also have an optimal protocol for Interacting Sensing. However, Hamiltonian reshaping comes with a approximation error, as we discussed above.
In \cref{app:reshaping_error}, we quantify how this error propagates into the estimation of the target parameter $q$ in the Interacting Sensing problem. Specifically, we examine the fate of our ideal estimator $\hat{q}_{\text{est}}$ when the Trotterized unitary is used in place of the ideal evolution. Using a perturbative expansion around the ideal evolution, we show that the resulting error in the estimation process depends on the number of Trotter steps $L$, the system size $n$, the operator norm of the Hamiltonian $\lambda$, and the optimal protocol parameters.

As detailed in \cref{app:reshaping_error}, we derive the following error bound for our estimator, which is now generally biased for any given realization our approximate evolution:
\begin{equation}
\mathbb{E}\left[\braket{\delta \hat{q}^2}\right] \lesssim \frac{\norm{\bs{a}}_1^2}{4t^2}\left\{1 +  C'\left(\frac{n\|\bs{a}\|_0 \lambda^2t^2}{L}\right)\right\}\label{eq:amse}
\end{equation}
for constant $C'$. Here, $\bs{a}$ is the optimal vector used in the protocol, and $\mathbb{E}[\cdot]$ is the expectation over all random sampling in our reshaping procedure. These expression quantifies how simulation error degrades the estimation performance if $L$ is too small.

To ensure this expression is not dominated by the reshaping error (i.e. the second term in brackets on the right in Eq.~\eqref{eq:amse}), we must choose the number of Trotter steps $L$ to satisfy
\begin{equation}
L \gtrsim C' n \|\bs{a}\|_0 \lambda^2 t^2.
\end{equation}
This guarantees that reshaping errors remain subleading and that the optimal precision for Diagonal Learning remains achievable for the Interacting Sensing problem, even with this approximate implementation.

\textit{The role of entanglement.---}In general, one can ask whether entanglement provides an advantage at different stages of the experiment. We focus on three key stages: (1) the initial state, which may be either a product state or an entangled state; (2) the control during the evolution, which can range from no control to single-qubit control or multi-qubit (entangling) control; and (3) the final measurement, which can be either local (single-qubit) or non-local (multi-qubit readout). The optimal protocol we consider begins with an entangled initial state, and applies multi-qubit control during evolution at each run. In \cref{app:alternative-protocols}, we consider two alternative protocols: one that uses a product initial state, single-qubit controls during evolution, and single-qubit measurements, and another that also uses a product initial state and single-qubit controls, but allows for entangled measurements at the end. As expected, both protocols generically perform worse than the optimal entangled protocols, indicating the crucial role of entanglement for achieving optimal scaling.

\textit{Conclusion.---}We have developed a general framework for the optimal estimation of linear functions of parameters encoded in \emph{any} diagonal many-body quantum Hamiltonian, and applied this to the problem of sensing a linear combination of local fields in the presence of \emph{arbitrary} interactions. Our results highlight several open directions. A key question is whether similar precision scaling can be achieved using only local operations—specifically, single-qubit control and measurement—without requiring global entanglement. Another direction is to extend the framework to parameters coupled to non-commuting generators. It is also important to explore efficient approximate methods for the underlying $\ell_1$-minimization problem \cite{tibshirani2012strong, rauhut2010compressive}. Addressing these questions could significantly broaden the applicability of optimal quantum sensing strategies to platforms subject to uncontrolled interactions, and inform further understanding of sensing limits for these problems.

\textit{Acknowledgments.---}We thank Youcef Bamaara, Soonwon Choi, Wenjie Gong, Raphael Kaubruegger, Tudor Manole, Daniel K. Mark, Ana Maria Rey, Yu Tong, and Peter Zoller for insightful discussions. 
This work was supported in part by ONR MURI, AFOSR MURI, NSF QLCI (award No.~OMA-2120757), DoE ASCR Quantum Testbed Pathfinder program (awards No.~DE-SC0019040 and No.~DE-SC0024220), NSF STAQ program, DARPA SAVaNT ADVENT, ARL (W911NF-24-2-0107), and NQVL:QSTD:Pilot:FTL. We also acknowledge support from the U.S.~Department of Energy, Office of Science, National Quantum Information Science Research Centers, Quantum Systems Accelerator (QSA) and from the U.S.~Department of Energy, Office of Science, Accelerated Research in Quantum Computing, Fundamental Algorithmic Research toward Quantum Utility (FAR-Qu).
 S.R.M. is supported by the
NSF QLCI (award No. OMA-2120757).
A.J.B.\@ acknowledges support from the NRC Research Associateship Program at NIST. Y.-X.W.~acknowledges support from a QuICS Hartree Postdoctoral Fellowship. 
J.B. notes that the views expressed in this article are those of the author and do not reflect the official policy or position of the U.S. Naval Academy, Department of the Navy, the Department of Defense, or the U.S. Government.

\bibliography{ref}

\newpage
\appendix
\onecolumngrid

\section{Optimal Protocol for Diagonal Learning}\label{app:opt_protocol}

In this appendix, we derive the Cram\'er-Rao bound for the Diagonal Learning problem, and demonstrate its saturability. This bound also provides a bound for the Interacting Sensing problem, after our reshaping procedure to diagonalize the corresponding interaction Hamiltonian. We reformulate the problem in a slightly more general way for arbitrary commuting and independent generators, instead of the specific choice of stabilizers used in the main text; however, we also note that such a problem can always be mapped to a qubit system and encoded in terms of these stabilizers.

\subsection{Derivation of Cram\'er Rao Bound}\label{sub_app:qcrb}
Here, we provide our proof of the Cram\'er-Rao Bound for learning a linear combination of parameters in any diagonal Hamiltonian, as a slight reformulation of the Diagonal Learning problem.

Consider the arbitrary Hamiltonian
\begin{equation}
 \hat{H}(s) = \sum_{j=1}^{m}\theta_j \hat{g}_j + \hat{H}_c(s),
\end{equation}
where $\hat{g}_j$ are \emph{traceless} generators for the Hamiltonian in an $N$-dimensional state space, and the task is to optimally learn the linear function $q= \bs{\alpha}\cdot \bs \theta$. We note that we may always redefine the generators so they are traceless, since any component of the Hamiltonian proportional to the identity is not learnable. For generality, we do not yet require mutual commutativity of the $\hat{g}_j$.

The general approach in previous works is to use the tightest \emph{single-parameter} quantum Cram\'er-Rao bound for the quantity $q=\bs\alpha\cdot\bs{\theta} = \sum_{j=1}^m \alpha_j\theta_j$ we wish to learn, under the assumption that the remaining $m-1$ unknown degrees of freedom in the problem are fixed~\cite{eldredge_optimal_2018}. As fixing these extra degrees of freedom requires extra information we do not have, this approach only provides a lower bound for the minimum attainable sensitivity. Thus, by minimizing this bound over \emph{all} possible choices for fixing these extra degrees of freedom, we can obtain the tightest possible lower bound, whose saturability is guaranteed by the existence of a protocol that achieves the corresponding measurement sensitivity. In essence, we are optimizing over all possible single-parameter bounds to place a lower bound on what is, fundamentally, a multi-parameter estimation problem.

To find a single-parameter bound on $q$, we rewrite the encoding Hamiltonian as
\begin{align}\label{eq:H_new_basis}
\hat{H}(s) &=\sum_{j=1}^m ({\bs{\alpha}}^{(j)}\cdot{\bs{\theta}})({\bs{\beta}}^{(j)}\cdot\hat{\bs{g}})+\hat{H}_c(s),
\end{align}
where $\{{\bs{\alpha}}^{(j)}\}_{j=1}^m$ are some choice of basis vectors such that ${\bs{\alpha}}^{(1)}={\bs{\alpha}}$ and $\{{\bs{\beta}}^{(j)}\}_{j=1}^m$ forms a dual basis such that ${\bs{\alpha}}^{(i)}\cdot{\bs{\beta}}^{(j)}=\delta_{ij}$. This set of vectors denotes a change of basis ${\bs{\theta}}\rightarrow\bs{q}$, where $q_j:={\bs{\alpha}}^{(j)}\cdot{\bs{\theta}}$ and $q_1=q$, our function of interest.

For simplicity of notation, let ${\bs{\beta}}:={\bs{\beta}}^{(1)}$. Then, we can define a $\bs{\beta}$-dependent generator of translations for the quantity of interest $q$ as
\begin{equation}\label{eq:generator}
\hat{g}_{q,{\bs{\beta}}}:=\frac{\partial \hat{H}}{\partial{q}}\Bigg|_{q_2,\cdots,q_N}={\bs{\beta}}\cdot\hat{\bs{g}},
\end{equation}
which follows from differentiating \cref{eq:H_new_basis}, assuming $q_j$ for $j > 1$ are fixed. Defining $\delta\hat{q} := \hat{q}_{\rm est} - q$ for unbiased estimator $\hat{q}_{\rm est}$ for $q$, we can then define the tightest (optimized over $\bs{\beta}$) single-parameter quantum Cram\'er-Rao bound for the function $q={\bs\alpha}\cdot{\bs\theta}$ as
\begin{align}\label{eq:QCRB_general}
\braket{\delta \hat{q}^2}\geq \max_{{\bs{\beta}}} \frac{1}{t^2\mathcal{F}(q|{\bs{\beta}})} \geq \max_{{\bs{\beta}}}\frac{1}{t^2\norm{\hat{g}_{q,{\bs{\beta}}}}_s^2},
\end{align}
where $\mathcal{F}(q|{\bs{\beta}})$ is the quantum Fisher information with respect to the parameter $q$, given a choice of ${\bs{\beta}}$ that fixes the subspace spanned by the coefficient vectors ${\bs{\beta}}^{(j)}$ with $j>1$ (specifying the extra $m-1$ degrees of freedom). This in turn is bounded by the squared seminorm of the generator $\hat{g}_{q,{\bs{\beta}}}$, which we take to mean the difference between the maximum and minimum eigenvalues of $\hat{g}_{q,{\bs{\beta}}}$ in \cref{eq:generator}~\cite{boixo2007generalized}. The optimization to find the tightest bound is done with respect to ${\bs{\beta}}$, subject to ${\bs{\beta}}\cdot{\bs{\alpha}}=1$.

Thus, we can reduce the problem of finding the bound to the following optimization problem:
\begin{align}\label{eq:minprob_app1}
 \text{minimize (w.r.t. ${\bs\beta}$): }&\norm{{\bs\beta} \cdot \hat{\bs g}}_{s}, \nonumber\\
 \text{subject to: }&\, {\bs\alpha}\cdot{\bs\beta}=1.
\end{align}

Now assume that the generators commute and are linearly independent, in addition to being traceless. Then there exists a basis $\{\ket{k}\}_{k=1}^N$ in which all generators are diagonal:
\begin{equation}
 \sum_{j=1}^{m}\theta_j \hat{g}_j = \sum_{k=1}^N \lambda_k \dyad{k}, \quad \bs \lambda = h^T \bs \theta,
\end{equation}
where $h \in \mathbb{R}^{m\times N}$ has rows given by the eigenvalues of the generators in this basis, i.e.~$h_{j,k} := \braket{k|\hat{g}_j|k}$. Since the generators are independent, the rows of $h$ are linearly independent, and $m \leq N$. Also note that via traceless-ness of the generators, each row of $h$ sums to $0$.
In this basis, the problem in Eq.~\eqref{eq:minprob_app1} becomes the linear program
\begin{align}
 \text{minimize (w.r.t. $\bs{\beta}\in \mathbb{R}^m$, $u,v\in \mathbb{R}$): }&\, v-u, \nonumber\\
 \text{subject to: }&\, u \leq (h^T\, \bs{\beta})_k \leq v \quad \forall k \in [N] , \nonumber\\
 & \, \bs{\alpha}^T \bs{\beta} = 1, 
\end{align}
The dual of this linear program is \cite{boyd2004convex}
\begin{align}\label{eq:dual_minprob_app1}
 \text{maximize (w.r.t. $\xi \in \mathbb{R},\bs{w},\bs{z} \in \mathbb{R}^N$): } &\, \xi, \nonumber\\
 \text{subject to: } & \,h(\bs{w}-\bs{z}) = \xi\bs\alpha, \nonumber \\
 & \sum_k w_k = \sum_k z_k = 1, \nonumber\\
 & w_k,z_k\geq 0.
\end{align}
By strong duality for linear programs, the answer to the problems in Eqs.~(\ref{eq:minprob_app1}) and (\ref{eq:dual_minprob_app1}) are equal.

Now, defining $\bs{a} := (\bs{w}-\bs{z})/\xi$, we have via the triangle inequality $\xi\norm{\bs{a}}_1 = \norm{\bs{w} - \bs{z}}_1 
\leq 2$. Thus, we can define the problem
\begin{align}\label{eq:1norm_minprob}
 \text{maximize (w.r.t. ${\bs{a} \in \mathbb{R}^N}$): } &\, \frac{2}{\|\bs{a}\|_1}, \nonumber\\
 \text{subject to: } & h\bs{a} = \bs\alpha,\nonumber\\
 & \sum_k a_k = 0,
\end{align}
whose solution will strictly upper bound that in Eq.~\eqref{eq:dual_minprob_app1}. Hence, the Cram{\'e}r-Rao bound for estimating $q$ may be written as
\begin{equation}\label{eq:Bnd_app}
\braket{\delta\hat{q}^2}\geq \max_{{\bs{\beta}}}\frac{1}{t^2\norm{\hat{g}_{q,{\bs{\beta}}}}_s^2}= \min_{\bs{a}\in\mathcal{A}}\frac{\norm{\bs{a}}_1^2}{4t^2},
\end{equation}
where $\mathcal{A}:=\{\bs a | h \bs a = \bs\alpha, \sum_k a_k = 0\}$.

We note that we can reformulate Eq.~\eqref{eq:1norm_minprob} as follows. Let us define the augmented constraint matrix
\begin{align}
    h' := \begin{bmatrix}
 \bs{1}^T \\ 
 h\\
 \tilde{h}
 \end{bmatrix},\qquad \bs{\alpha}' := \begin{bmatrix}
 0 \\ 
 \bs{\alpha}\\
 \tilde{\bs{\alpha}}
\end{bmatrix}
\end{align}
where $\tilde{h}\in \mathbb{R}^{(N-m-1)\times N}$ is any matrix whose rows are independent and span the orthogonal complement of the rowspace of $\begin{bmatrix}
 \bs{1}^T \\ 
 h
\end{bmatrix}$, and $\tilde{\bs{\alpha}}\in \mathbb{R}^{N-m-1}$ parametrizes these unconstrained dimensions. $h'$ is thus a full rank matrix in $\mathbb{R}^{N\times N}$, so every feasible solution $\bs{a}$ to Eq.~\eqref{eq:1norm_minprob} can be uniquely written as $\bs{a} = (h')^{-1}\bs{\alpha}'$. We then have the equivalent linear program
\begin{align}\label{eq:1norm_minprob_aug}
 \text{maximize (w.r.t. ${\tilde{\bs{\alpha}} \in \mathbb{R}^{N-m-1}}$): } &\, \frac{2}{\norm{(h')^{-1}\bs{\alpha}'}_1},
\end{align}
as used in the main text for our specific choice of stabilizers.

As interpreted there, we can form a complete set of mutually commuting, independent generators $\hat{g}'_j = \sum_{k=1}^N h'_{j,k}\ket{k}\bra{k}$, such that $\hat{g}'_0 = \hat{I}$ is the $N$-dimensional identity operator, and $\hat{g}'_j = \hat{g}_j$ for $1\leq j \leq m$. Defining $\theta' = (\theta_0,\bs{\theta}^T,\bs{0}^T) \in \mathbb{R}^N$, our problem is now equivalent to learning the function $q = \bs{\alpha}'\cdot \bs{\theta}'$, with generating Hamiltonian
\begin{align}
    \hat{H}'(s) = \sum_{j=0}^{N-1} \theta'_j \hat{g}'_j + \hat{H}_c(s),
\end{align}
\emph{with the additional stipulation} that we know a priori that $\theta'_j=0$ for $j > m$. Thus, the corresponding entries of $\bs{\alpha}'$ are unconstrained. We also note that since the parameter coupled to the identity cannot be estimated, its coefficient $\alpha'_0$ must vanish.

While fixing the unconstrained elements $\tilde{\bs{\alpha}}$ thus has no physical consequence for our ability to estimate $q$, it \emph{does} affect our presented protocol insofar as it depends on the values of $\bs{a} = (h')^{-1} \bs{\alpha}'$ (see App.~\ref{sub_app:protocol} below).

\subsection{Extension of Optimal Protocol to Arbitrary Diagonal Hamiltonians}\label{sub_app:protocol}

In the main text, we presented an optimal protocol in the case that our generators correspond to some subset of the non-trivial stabilizers in $\mathcal{S}$. We briefly note how to generalize our approach to the arbitrary form of the diagonal Hamiltonian in the App.~\ref{sub_app:qcrb}, assuming traceless, linearly independent, commuting generators. Given an optimal solution $\bs{a}\in \mathcal{A}$ that solves Eq.~\eqref{eq:1norm_minprob}, we form sets $\Lambda^+ := \{k|a_k > 0\}$ and $\Lambda^- := \{k|a_k < 0\}$ of our diagonal basis state labels, and form enumerations of these sets $k^\pm(\mu)\in\Lambda^\pm$ for $1 \leq \mu \leq N^\pm$ for $N^\pm := |\Lambda^\pm|$. We also define the superposition states $|\psi_{k,k'}(\theta)\rangle = (|k\rangle + e^{i\theta}|k'\rangle )\sqrt{2}$, and consider coherent basis switching operators $\hat{S}_{k,k''}$ ($\hat{S}_{k',k''})$ that take $\ket{k}$ ($\ket{k'}$) to $\ket{k''}$ in $|\psi_{k,k'}$, without affecting the the basis state in the other half of the superposition, e.g.~$\hat{S}_{k,k''} = \sum_{k''\neq k,k'} |k''\rangle \langle k''| + |k\rangle \langle k'| + |k'\rangle \langle k|$.

Now, starting from $\ket{\psi_{k^+(1),k^-(1)}(0)}$, the optimal protocol largely follows the stabilizer version in the main text, with the appropriate evolution durations between coherent basis switching operations given by $\Delta t^\pm_{\mu} = 2t(a_{k_\mu}/\norm{\bs{a}}_1)$ for the left/right halves of the superposition. Owing to the constraint $\sum_k a_k=0$ in Eq.~\eqref{eq:1norm_minprob} for any choice of $\bs{a}$, we are again guaranteed that $\sum_{\mu} a_{k^\pm(\mu)} = \pm\norm{\bs{a}}_1/2$, and so both halves of the superposition evolve for equal total durations: $\sum_{\mu} \Delta t_{\mu}^+ = \sum_{\mu} \Delta t_{\mu}^- = t$.

The final evolved state is $\ket{\psi_{k^+(N^+),k^-(N^-)}(\varphi)}$ for accumulated phase difference $\varphi = 2tq/\norm{\bs{a}}_1$. For readout, we can neglect qubits whose state is identical between the left and right halves of the superposition, since these are disentangled from the rest of the system. The relative phase for the remaining qubits can be read out via a standard parity measurement; alternatively, the state of these qubits can be disentangled by a series of applied CNOT and CNOT$_0$ gates to transfer the relative phase to a single qubit, which can then be measured. In either scenario, this amounts to effectively measuring the observable $\hat{O} = -i(|k^+(N^+)\rangle \langle k^-(N^-)| - |k^-(N^-)\rangle \langle k^+(N^+)|)$ over the whole system. We then have
\begin{align}
 \langle \hat{O}\rangle = \sin\left(\frac{2tq}{\norm{\bs{a}}_1}\right),\qquad \langle (\Delta\hat{O})^2\rangle = \cos\left(\frac{2tq}{\norm{\bs{a}}_1}\right)^2,
\end{align}
for $\Delta \hat{O} = \hat{O} - \braket{\hat{O}}$, where expectations $\braket{\cdot}$ are taken with respect to $\ket{\psi_{k^+(N^+),k^-(N^-)}(\varphi)}$. 
Thus, for phase sensitivity $\braket{\delta \hat{q}^2} = \langle (\Delta\hat{O})^2\rangle/|\partial \langle\hat{O}\rangle/\partial q|^2$, we obtain $\braket{\delta \hat{q}^2} = \norm{\bs{a}}_1^2/4t^2$---so long as $\braket{\hat{O}}$ continues to enable an unbiased estimate of $q$ without phase slips. The existence of such a protocol guarantees that the bound in Eq.~\eqref{eq:Bnd_app} is indeed tight, and that the solutions to Eq.~\eqref{eq:dual_minprob_app1} and Eq.~\eqref{eq:1norm_minprob} are, in fact, equal.

\section{Sparsity of the Optimal Solution}\label{app:sparsity}
We address the structure of the optimal vector $\bs a$. Here, we show that the optimal $\bs a$ is \emph{sparse} when the number of generators is small:
\begin{theorem}[Sparsity]
Let $m=\mathrm{poly}(n)$ be the number of generators on $n$ qubits. Then there exists an optimal solution $\bs a$ to the minimization problem~\eqref{eq:Bnd_app} with support size $\|\bs a\|_0 = \mathrm{poly}(n)$.
\end{theorem}
\begin{proof}[Proof sketch]
This result follows from classical linear-programming theory on basic feasible solutions \cite{bertsimas1997introduction, chvatal1983linear, boyd2004convex}. 
For completeness, we provide the steps of the argument. We rewrite the problem in Eq.~(\ref{eq:1norm_minprob}) in standard linear-programming form. 
Write $\bs a = \bs a^+ - \bs a^-$ with $\bs a^+,\bs a^-\ge0$. Then
\begin{equation}
\|\bs a\|_1 = \mathbf{1}^\top(\bs a^+ + \bs a^-), \qquad 
\begin{bmatrix}
 \bs{1}^T \\ 
 h
\end{bmatrix}(\bs a^+ - \bs a^-) = \begin{bmatrix}
 0 \\ 
 \bs \alpha
\end{bmatrix}.
\end{equation}
Introducing $\bs x=\begin{bmatrix}
 \bs a^+ \\ 
 \bs a^-
\end{bmatrix}\in\mathbb R^{2N}$, $\bs c=\begin{bmatrix}
 \bs{1} \\ 
 \bs 1
\end{bmatrix}$, and 
$A=\begin{bmatrix}
 \bs{1}^T & -\bs{1}^T\\ 
 h & -h
\end{bmatrix}\in\mathbb R^{m+1\times 2N}$, the problem is equivalent to
\begin{equation}\label{eq:1norm_min_lp}
 \min_{\bs x\ge0}\; \bs c^\top \bs x 
 \quad \text{subject to}\quad A \bs x = \begin{bmatrix}
 0 \\ 
 \bs \alpha
\end{bmatrix}.
\end{equation}
\emph{Step 1.} 
The feasible set $P=\{\bs x\ge0 : A \bs x=\begin{bmatrix}
 0 \\ 
 \bs \alpha
\end{bmatrix}\}$ is a polyhedron. Since the objective is linear, there exists an optimal solution at an extreme point of $P$.\\
\emph{Step 2.} 
\emph{Basic feasible solutions} correspond to a subset of the extreme points of $P$. 
Let $B=\{j: x_j>0\}$ denote the support of a basic feasible solution $\bs x$. Then the columns $A_B$ must be linearly independent: otherwise, there exists a nonzero $\bs d_B$ with $A_B \bs d_B=0$, and small perturbations $\bs x_B\pm \varepsilon \bs d_B$ remain feasible and nonnegative, contradicting extremality. Thus
\begin{equation}
 |B| \le \mathrm{rank}(A).
\end{equation}
Because $\mathrm{rank}(A)=\mathrm{rank}(\begin{bmatrix}
 \bs{1}^T \\ 
 h
\end{bmatrix})=:r$, every basic feasible solution has at most $r$ positive entries.\\
\emph{Step 3.} 
Returning to $\bs a=\bs a^+-\bs a^-$, note that, for each index $j$, at most one of $a^+_j,a^-_j$ can be strictly positive in an optimal solution. Indeed, if both are positive, we may simultaneously reduce them by their minimum, leaving $\bs a$ unchanged while strictly lowering the objective. Hence in any optimal solution, the number of nonzero entries of $\bs a$ is bounded by the number of positive entries of $\bs x$. Combining with Step~2,
\begin{equation}
 \|\bs a\|_0 \le |B| \le r.
\end{equation}
\emph{Step 4.} 
Since $r= m+1$, we conclude that there exists an optimal solution $\bs a$ with
\begin{equation}
 \|\bs a\|_0 \le m+1.
\end{equation}
In particular, for $m=\mathrm{poly}(n)$ the optimal $\bs a$ has support size polynomial in $n$, as claimed.
\end{proof}
This sparsity result implies that, even for an exponentially large Hilbert space (of dimension $N$), the optimal solutions to the Diagonal Learning bound are supported on only $m$ entries. For our optimal protocol, if $m=\mathrm{poly}(n)$, this implies that we only require $\mathrm{poly}(n)$ coherent basis switching operations to optimally encode our phase.

\section{Hamiltonian Reshaping Error}\label{app:reshaping_error}
In this appendix, we analyze the errors associated with the proposed random Trotterization protocol for our Hamiltonian reshaping.
While we focus on the case of qubit Hamiltonians, all results in this section directly extend to the case with multi-qudit Hamiltonians (and more generally, Hamiltonians on any finite-dimensional systems) using concentration bounds on random product formulas in Hamiltonian simulation~\cite{chen_concentration_2021}. For the qudit case, all error bounds will acquire an extra prefactor that depends on dimensions of individual sensors. We also comment that, as in the stabilizer case, alternative methods for achieving the desired cancellation of off-diagonal terms exist; for instance, one may utilize a generalization of the proposed incommensurate strong field protocol, where strong fields are locally applied in the basis in which our generators our diagonalized.

In App.~\ref{app:c_1}, we review expected error bounds for effective unitary evolution via randomized Trotterization, adding a calculation of the expected variance for this error. In App.~\ref{app:reshaphing_error_opt_protocol}, we then derive specific instances of these error bounds as applied to our proposed optimal protocol. Finally, in App.~\ref{app:c_3}, we use these results to compute an upper bound on the added noise in our estimation of $q$ arising from this approximate evolution. This error bound leads to a requirement on the needed step size for our protocol to remain asymptotically optimal.

\subsection{Reshaping Error Bounds for the Effective Hamiltonian}
\label{app:c_1}
Here, we analyze the error associated with simulating the effective Hamiltonian $\hat{H}_\text{eff}$ using randomized Hamiltonian simulation techniques. Our goal is to simulate time evolution under a reshaped Hamiltonian composed entirely of commuting terms, defined as
\begin{equation}
 \hat{H}_\text{eff} = \frac{1}{N} \sum_{k=1}^N \hat{s}_k \hat{H}_0 \hat{s}_k,
\end{equation}
where $\hat{H}_0$ is an arbitrary $n$-qubit Hamiltonian, and each $\hat{s}_k$ is an element of the stabilizer group $\mathcal{S} = \{ \hat{I}, \hat{Z} \}^{\otimes n}$. This form arises from the stabilizer-averaging technique described in the main text, which projects $\hat{H}_0$ onto a diagonal subspace by randomizing over $\mathcal{S}$. The resulting $\hat{H}_\text{eff}$ contains only commuting Pauli-$Z$ terms and is suitable for efficient learning and sensing.

To simulate the time evolution under $\hat{H}_\text{eff}$ up to a total time $t$, we apply a QDRIFT-like  method~\cite{campbell_random_2019}, which samples random Hamiltonian terms to construct an approximate evolution. We define
\begin{equation}\label{eq:single_timestep_unitary}
 \hat{V}_k = e^{-i \hat{s}_k \hat{H}_0 \hat{s}_k \, \Delta t}, \quad \Delta t = \frac{t}{L},
\end{equation}
where each $\hat{s}_k$ is drawn uniformly at random from $\mathcal{S}$, and $L$ is the total number of steps. The simulated evolution operator is then the product
\begin{equation}
 \hat{\tilde{U}} = \hat{V}_L \cdots \hat{V}_1,
\end{equation}
which serves as an approximation to the ideal evolution operator
\[
\hat{U} = e^{-i \hat{H}_\text{eff} t}.
\]

To quantify the simulation error, we consider the operator-norm distance between the true evolution and its randomized approximation. We can split this distance into a deterministic error and a random error as
\begin{equation}\label{eq:true_vs_random_unitary_error}
 \norm{\hat{U}-\hat{\tilde{U}}} \leq \norm{\hat{U}-\mathbb{E}[\hat{V}]^L} + X, \quad X = \norm{ \hat{V}_L \cdots \hat{V}_1 - \mathbb{E}[\hat{V}]^L },
\end{equation}
where $\mathbb{E}[\hat{V}]$ denotes the average of $\hat{V}_k$ in Eq.~(\ref{eq:single_timestep_unitary}) over all $
\hat{s}_k$ in the stabilizer group. The error term $\|\hat{U}-\mathbb{E}[\hat{V}]^L\|$ is deterministic and does not depend on the random choice of the unitaries. Using the result from Ref.~\cite{campbell_random_2019}, we have 
\begin{equation}\label{eq:bias_error}
 \norm{\hat{U}-\mathbb{E}[\hat{V}]^L} \leq \frac{ 2\lambda^2t^2}{L},
\end{equation}
where $\lambda = \| \hat{H}_0 \|$ is the operator norm of the original Hamiltonian. In our setting, where $\hat{H}_0$ contains at most poly($n$) Pauli terms, we can bound $\lambda = O(\mathrm{poly}(n))$. The other term $X$ in Eq.~(\ref{eq:true_vs_random_unitary_error}) is a random variable whose value depends on the random sample trajectory $\{ \hat{s}_k \}_{k=1}^L$. The expectation $\mathbb{E}[\hat{V}]^L$ is taken over all such trajectories, treating each $\hat{V}_k$ as an independent random unitary.

Using concentration results from randomized product formulas, specifically Proposition 3.3 from Ref.~\cite{chen_concentration_2021}, we obtain a tail bound on the error for $X$. For any $\epsilon \in [0, \infty)$, the probability that $X$ exceeds $\epsilon$ is bounded as
\begin{equation}
 \mathrm{Pr}(X \geq \epsilon) \leq f(\epsilon), \quad f(\epsilon) = 2N \exp\left(-\frac{L \epsilon^2}{8t^2\lambda^2 + \frac{4t\lambda \epsilon}{3}} \right).
\end{equation}
Since $X$ is bounded from above by $1$, we can tighten the tail bound as
\begin{equation}
 \mathrm{Pr}(X \geq \epsilon) \leq 
 \begin{cases}
 f(\epsilon), & \epsilon \leq 1, \\
 0, & \text{otherwise}.
 \end{cases}        
\end{equation}
Next, we compute the expectation and variance of the random variable $X$:
\begin{gather}\label{eq:X-expectation}
  \mathbb{E}[X] = \int_0^\infty \mathrm{Pr}(X\geq \epsilon) \,d\epsilon \leq \int_0^1 \min\{1, f(\epsilon)\}\, d\epsilon.
\\\label{eq:X_variance}
  \mathrm{Var}[X] \leq \mathbb{E}[X^2] = \int_0^\infty \mathrm{Pr}(X^2\geq \tau )  \,d\tau \leq \int_0^1 2\epsilon \, \min\{1, f(\epsilon)\} \, d\epsilon.
\end{gather}
Following Ref.~\cite{chen_concentration_2021}, note that the term $\min\{1, f(\epsilon)\} $ remains equal to 1 for $\epsilon \lesssim \max (n \lambda t/L, \sqrt{n \lambda^2 t^2/L} )$. For larger $\epsilon$, the integrand decays exponentially until it vanishes for $\epsilon > 1$, and the corresponding integral over this range yields a contribution of order $O(\max (\lambda t/L, \sqrt{ \lambda^2 t^2/L} ))$ in Eq.~(\ref{eq:X-expectation}) and a contribution of order $O(\max (\lambda^2 t^2/L^2, {\lambda^2 t^2/L} ))$ in Eq.~(\ref{eq:X_variance}). Thus, we obtain the following bounds:
\begin{equation}\label{eq:expectaion_variance_X}
 \mathbb{E}[X]
 \lesssim c\sqrt{\frac{n\lambda^2 t^2}{L}},\quad \mathrm{Var}[X] 
 \lesssim c' \frac{n \lambda^2 t^2}{L},
\end{equation}
for constants $c, c'$, where $\lesssim$ denotes the bound asymptotically in the limit of large $L$.

Note that the bound for the expectation value of $X$ was already introduced in Ref.~\cite{chen_concentration_2021}. Here, we additionally derived the expression for its variance. These bounds show that both the average error and the standard deviation of $X$ decay as $1/\sqrt{L}$ asymptotically in $L$. As we discuss in App.~\ref{app:reshaphing_error_opt_protocol} below, this determines the precision of parameter estimation in our protocol. In particular, the concentration of $X$ around zero ensures that the simulated dynamics remain close to those generated by the ideal commuting Hamiltonian, with high probability.

\subsection{Error Bounds for the Full Controlled Protocol with Trotterized Simulation}\label{app:reshaphing_error_opt_protocol}

We now analyze the error introduced by the reshaping of the Hamiltonian in the optimal protocol for estimating $q$. In the ideal case, the full unitary evolution for the optimal protocol (involving $R = \|\bs{a}\|_0$ rounds of coherent controls) is given by
\begin{equation}
 \hat{U}_1 = \prod_{r=1}^R \hat{U}_{\text{eff}}(t_r - t_{r-1}) \hat{W}_r,
\end{equation}
where
\begin{equation}
 \hat{U}_{\text{eff}}(t) = e^{-i\hat{H}_{\text{eff}} t},
\end{equation}
and $\hat{W}_r$ are the control unitaries (e.g., coherent switching gates between different GHZ-like superpositions) applied at times $t_0 = 0 < t_1 < \cdots < t_R = t$.

We now define an approximate protocol where the evolution under $\hat{H}_{\text{eff}}$ is replaced by a randomized Trotter expansion using $L \gg R$ steps. Specifically, for each interval $[t_{r-1}, t_r]$, we apply a sequence of sampled unitaries $\{ \hat{V}_i \}$ defined as in Eq.~(\ref{eq:single_timestep_unitary}). The sensor evolution is then
\begin{equation}
 \hat{U}_2 = \prod_{r=1}^R \left( \prod_{i = (L t_{r-1}/t) + 1}^{L t_r/t} \hat{V}_i \right) \hat{W}_r,
\end{equation}
which interleaves control gates $\hat{W}_r$ with randomized simulation of each evolution segment. Here, we choose $t_r$ to be integer multiples of $t/L$ and assume that $\min_r (t_r-t_{r-1}) \gg t/L$.

We consider the operator norm distance between the ideal and the implemented unitaries:
\begin{align}\label{eq:protocol_reshaping_error}
 Y = \norm{ \hat{U}_1 - \hat{U}_2 } &\leq \sum_{r=1}^R \norm{ \hat{U}_{\text{eff}}(\delta t_r) - \prod_{i = L t_{r-1}/t + 1}^{L t_r/t} \hat{V}_i } \nonumber \\
 &\leq \sum_{r=1}^R \norm{ \hat{U}_{\text{eff}}(\delta t_r) - \mathbb{E}[\hat{V}(\delta t_r)] } + \sum_{r=1}^R \norm{ \mathbb{E}[\hat{V}(\delta t_r)] - \prod_{i = L t_{r-1}/t + 1}^{L t_r/t} \hat{V}_i },
\end{align}
where $\delta t_r = t_r - t_{r-1}$, and the expectation $\mathbb{E}[\hat{V}(\delta t_r)]$ denotes the averaged evolution over the interval. Next, using Eq.~(\ref{eq:bias_error}), the deterministic error between the average and ideal evolution satisfies
\begin{equation}
 \left\| \hat{U}_{\text{eff}}(\delta t_r) - \mathbb{E}[\hat{V}(\delta t_r)] \right\| \leq \frac{2 \lambda^2 \delta t_r^2}{L \delta t_r / t} = \frac{2 \lambda^2 t \delta t_r}{L}.
\end{equation}
Summing over all $r$, we obtain
\begin{equation}
 \sum_{r=1}^R \left\| \hat{U}_{\text{eff}}(\delta t_r) - \mathbb{E}[\hat{V}(\delta t_r)] \right\| \leq \frac{2 \lambda^2 t}{L} \sum_{r=1}^R \delta t_r = \frac{2 \lambda^2 t^2}{L}.
\end{equation}
The remaining term in Eq.~(\ref{eq:protocol_reshaping_error}) is random and dependent on the choice of sampled unitaries. Since each $\hat{V}_i$ is independently sampled, by using Eq.~(\ref{eq:expectaion_variance_X}), we obtain the following bound on the expectation of $Y$:
\begin{align}\label{eq:expectation_Y}
 \mathbb{E}[Y] &\lesssim \frac{2 \lambda^2 t^2}{L} + c \sum_{r=1}^R \sqrt{ \frac{n \lambda^2 \delta t_r^2}{L \delta t_r / t} } 
 = \frac{2 \lambda^2 t^2}{L} + c \sqrt{ \frac{n \lambda^2 t}{L} } \sum_{r=1}^R \sqrt{\delta t_r} \lesssim 2c\sqrt{ \frac{nR \lambda^2 t^2}{L} },
\end{align}
using the fact that $\sum_{r=1}^R\sqrt{\delta t_r} \leq \sqrt{Rt}$ via an application of the Cauchy-Schwarz inequality.

\subsection{Reshaping Error Propagation to Parameter Estimation}
\label{app:c_3}
We now examine how the approximation error in the unitary evolution affects the precision of estimate for $q$ in our optimal protocol. Recall that, in the ideal setting, the protocol begins with a GHZ-like superposition state and applies a sequence of coherent switching operations interleaved with evolution under the effective Hamiltonian. The final state takes the form
\begin{equation}
 \ket{\psi_1} = \hat{U}_1 \ket{\psi(0)} = \frac{1}{\sqrt{2}} \left( \ket{x} + e^{i\varphi} \ket{y} \right),
\end{equation}
where $\ket{x}$ and $\ket{y}$ are two computational basis states, and the accumulated phase is given by
\begin{equation}
 \varphi = \frac{2q t}{\|\bs{a}\|_1},
\end{equation}
with $\bs{a}$ denoting the optimized Hadamard transformed vector in $\mathbb{R}^N$ (see Eq.~(\ref{eq:Bnd_app})).

To estimate $\varphi$, and thereby estimate $q$, we effectively measure the observable
\begin{equation}
 \hat{O} = -i \left(\dyad{x}{y} -\dyad{y}{x}\right),
\end{equation}
which generates interference between $\ket{x}$ and $\ket{y}$. In the limit where $\varphi$ is small (i.e., when $qt \ll 1$), we have
\begin{equation}
 \bra{\psi_1} \hat{O} \ket{\psi_1} \simeq \varphi,
\end{equation}
allowing direct extraction of $\varphi$ from measurement outcomes.

Now suppose the exact unitary $\hat{U}_1$ is replaced by the approximate unitary $\hat{U}_2$ defined by the Trotterized simulation protocol discussed in App.~(\ref{app:reshaphing_error_opt_protocol}). The final state becomes $\ket{\psi_2} = \hat{U}_2 \ket{\psi(0)}$, and the estimated phase is
\begin{align}
\varphi_{\text{approx}} &= \bra{\psi_2} \hat{O} \ket{\psi_2}= \bra{\psi_1} \hat{O} \ket{\psi_1} + \delta \varphi,
\end{align}
where $\delta \varphi$ captures the error due to the deviation $\hat{U}_2 - \hat{U}_1$ in the unitary evolution. Using H\"older's inequality $\|\hat{O}\delta\hat{\rho}\|_1 \leq \|\hat{O}\|\|\delta\hat{\rho}\|_1 $ for trace norm $\| \cdot \|_1 := \Tr[|\cdot|]$, where $\delta\hat{\rho} := \ket{\psi_2}\bra{\psi_2} - \ket{\psi_1}\bra{\psi_1}$, we can obtain the bound
\begin{equation}
\label{eq:delta.varphi.bound}
 |\delta \varphi| \leq 2\|\hat{O}\|\|\hat{U}_1-\hat{U}_2\|
 .
\end{equation}

To quantify the impact of this error on estimation, we compute the average mean-square error between our estimate and the actual value of the phase, averaged over both measurement outcomes and random sampling for our Trotterization. Specifically, for deviation $\delta \hat{O} = \hat{O} - \varphi$, we are interested in the quantity
\begin{align}
\mathbb{E}[\braket{\psi_2|(\delta\hat{O})^2|\psi_2}]= \braket{\psi_1|(\Delta\hat{O})^2|\psi_1} + \mathbb{E}\left[\Tr[\hat{O}^2\delta\hat{\rho}]\right] - 2\varphi\mathbb{E}[\delta\varphi],\label{eq:amse_supp}
\end{align}
where $\Delta\hat{O} = \hat{O} - \braket{\psi_1|\hat{O}|\psi_1}$, and we assume $\braket{\psi_1|\hat{O}|\psi_1} = \varphi$.

First, we note that the first term in Eq.~\eqref{eq:amse_supp} is simply the variance of $\hat{O}$ in the ideal protocol. Now, for the second term on the right in Eq.~\eqref{eq:amse_supp}, using using Eq.~(\ref{eq:expectation_Y}) we obtain
\begin{align}
 & \mathbb{E}\left[\Tr[\hat{O}^2\delta\hat{\rho}]\right] \leq 2 \|\hat{O}^2\|\,\mathbb{E}\left[ \| \hat{U}_1 - \hat{U}_2 \| \right] \lesssim 4c\sqrt{\frac{nR\lambda^2t^2}{L}}.
\end{align}
Lastly, the final term in Eq.~\eqref{eq:amse_supp} may be similarly bounded by noting
\begin{align}
    |\mathbb{E}[\delta \varphi]|\leq 
 2\| \hat{O} \| \, \mathbb{E}\left[ \| \hat{U}_1 - \hat{U}_2 \| \right] \lesssim 4c\sqrt{\frac{nR\lambda^2t^2}{L}},
\end{align}
and that $|\varphi| \leq \|\hat{O}\|$.

Now, relating $\varphi$ to $q$, we have the average mean-square error
\begin{align}
    \mathbb{E}[\braket{\psi_2|\delta \hat{q}^2|\psi_2}] \lesssim \frac{\norm{\bs{a}}_1^2}{4t^2}\left\{1\, + \, 12c \sqrt{\frac{nR\lambda^2t^2}{L}}\right\},
\end{align}
where $\delta \hat{q} = \delta\hat{O}\norm{\bs{a}}_1/(2t)$.
Thus, to maintain the optimal scaling of the bias and variance in estimating $q$---i.e., matching the variance of our estimate in the ideal case---we must ensure that the number of Trotter steps $L$ satisfies
\begin{equation}
 L \gtrsim c''n\lambda^2 t^2\norm{\bs{a}}_0,
\end{equation}
for constant $c''$, where we note that $\|\bs{a}\|_0 \leq m+1$. This condition guarantees that the simulation error remains subleading to the optimal sensitivity bound and does not degrade the estimation precision achieved by the optimal protocol.

\section{Universality of Diagonal Learning}\label{app:universality}
In this appendix, we show two specific instances where our bound reproduces previously known results: 1) non-interacting qubit sensor networks, and 2) networks of Mach-Zehnder interferometers. 

\subsection{Independent Stabilizers \label{app:indepstab}}
In Ref.~\cite{eldredge_optimal_2018}, the optimal bound for estimating the linear function $q=\bs{\alpha} \cdot \bs{\theta}$ of the Hamiltonian $\hat{H} = \sum_j \theta_j \hat{Z}_j$ is given by
\begin{equation}\label{eq:z_Bnd}
\braket{\delta\hat{q}^2} \geq \frac{\|\bs{\alpha}\|_\infty^2}{4t^2}.
\end{equation}
We now show that our bound in Eq.~(\ref{eq:Bnd_app}) reproduces this result for such Hamiltonians and further extends it to Hamiltonians with arbitrary independent generators drawn from an arbitrary $n$-qubit stabilizer group. By \emph{independent}, we mean a set of commuting Pauli operators $\{\hat{g}_j\}_{j=1}^n$ such that no nontrivial product of these operators is in the set. For an $n$-qubit system, there can be at most $n$ such independent generators, e.g. $\{\hat{Z_1},\cdots,\hat{Z_n}\}$ for the stabilizer group of pauli-Z strings. 

To get the bound in \cref{eq:z_Bnd}, consider the dual formulation of Eq.~(\ref{eq:1norm_min_lp}):
\begin{equation}
\max_{\bs{y}\in\mathbb{R}^{n+1}} \, \bs{y}^T \begin{bmatrix}
    0\\ \bs{\alpha}
\end{bmatrix}
\quad \text{subject to} \quad A^T \bs{y} \leq \bs{1},
\end{equation}
where $A=\begin{bmatrix}
 \bs{1}^T & -\bs{1}^T\\ 
 h & -h
\end{bmatrix} \in \mathbb{R}^{n+1\times 2N}$, and $h\in \mathbb{R}^{n\times N}$ comprises the $n$ relevant rows of the Hadamard matrix corresponding to the independent generator of the set $\{\hat{g}_j\}_{j=1}^n$. If the $n$ generators are independent, one can always select a column of $A$ whose entries align in sign with those of $\bs{y}$. In this case, the constraint $A^T \bs{y} \leq \bs{1}$ is equivalent to $\|\bs{y}\|_1 \leq 1$. The linear program therefore becomes
\begin{equation}
\max_{\bs{y}\in\mathbb{R}^{n+1}} \, \bs{y}^T \begin{bmatrix}
    0\\ \bs{\alpha}
\end{bmatrix}
\quad \text{subject to} \quad \|\bs{y}\|_1 \leq 1,
\end{equation}
whose maximum is $\|\bs{\alpha}\|_\infty$. Consequently, for $n$ independent generators, our bound becomes:
\begin{equation} \braket{\delta\hat{q}^2} \geq \min_{\bs{a}\in\mathcal{A}}\frac{\norm{\bs{a}}_1^2}{4t^2}=\frac{\|\bs{\alpha}\|_\infty^2}{4t^2}. \end{equation}
which exactly matches the bound of Ref.~\cite{eldredge_optimal_2018}.

\subsection{Phase Sensing for Bosonic Modes}
\label{app:boson-diaglearning}

In Ref.~\cite{bringewatt_optimal_2024}, the optimal bound for estimating the linear function $q=\bs{\alpha} \cdot \bs{\theta}$ of the bosonic Hamiltonian $\hat{H} = \sum_{j=1}^m \theta_j \hat{n}_j$ is given by
\begin{equation}\label{eq:mzi_bnd}
\braket{\delta\hat{q}^2} \geq \frac{\max\{\|\bs \alpha\|_{1,+}^2, \|\bs \alpha\|_{1,-}^2\}}{P^2t^2},
\end{equation}
where $\hat{n}_j$ represents the number operator of the bosonic mode $j$, $P$ is the total number of photons, and $\|\bs \alpha\|_{1,+}$($\|\bs \alpha\|_{1,-}$) denotes the $\ell_1$-norm restricted to positive (negative) entries of $\bs \alpha$. Similar to the previous section (\cref{app:indepstab}), we now show that our bound in Eq.~(\ref{eq:Bnd_app}) reproduces this result as well. To see this, again consider the dual formulation of Eq.~(\ref{eq:1norm_min_lp}):
\begin{equation}
\max_{\bs{y}\in\mathbb{R}^{m+1}} \, \bs{y}^T \begin{bmatrix}
    0\\ \bs{\alpha}
\end{bmatrix}
\quad \text{subject to} \quad A^T \bs{y} \leq \bs{1},
\end{equation}
where $A=\begin{bmatrix}
 \bs{1}^T & -\bs{1}^T\\ 
 h_b & -h_b
\end{bmatrix} \in \mathbb{R}^{m+1\times 2N}$. Here, $h_b$ is a matrix such that $\mathrm{diag}(h_b^T \bs \theta)$ is the diagonal representation of the Hamiltonian in the basis of number states. The reason for the finite dimension $N$ is that we fix the total number of maximum photons, i.e.\ $P$. Thus, the columns of the matrix $h_b$ are all the elements in the set $\mathcal{P}= \{ \bs{p} \in \mathbb{Z}_{\ge 0}^m : \sum_j p_j \le P \}$, and this means that $N=|\mathcal{P}|=\binom{P+m}{m}$. \\
To prove the bound in \cref{eq:mzi_bnd}, let $\bs{y} = \begin{bmatrix}
    w \\\bs{u}
\end{bmatrix}$.  
Then the constraint $A^T \bs{y} \le \bs{1}$ becomes
\begin{equation}\label{eq:ut_cond}
|\bs{p}^T \bs{u} + w| \le 1 
\quad \forall\, \bs{p} \in \mathcal{P},
\end{equation}
which is equivalent to $|w|\le 1$ and
\begin{equation} 
u_j \in \left[-\frac{1 + w}{P}, \frac{1 - w}{P}\right]\quad \forall i\in\{1, \cdots ,m\}.
\end{equation}
The dual objective is $\bs{u}^\top \bs{\alpha}$, which is maximized by choosing each $u_j$ at the endpoint of its feasible interval consistent with the sign of $\alpha_j$. i.e.
\begin{equation}
u_j^\star =
\begin{cases}
(1 - w)/P, & \alpha_j > 0, \\[4pt]
-(1 + w)/P, & \alpha_j < 0.
\end{cases}
\end{equation}
Hence, for fixed $w$,
\begin{equation}
\bs{u}^{*T} \bs{\alpha}
= \frac{1 - w}{P}\|\bs{\alpha}\|_{1,+}
+ \frac{1 + w}{P}\|\bs{\alpha}\|_{1,-}.
\end{equation}
Maximizing the objective over $w \in [-1,1]$ gives
\begin{equation}
\max_{w,\,\bs{u}}\,  \bs{u}^T \bs{\alpha}
= \frac{2}{P}\max\{\|\bs{\alpha}\|_{1,+},\,\|\bs{\alpha}\|_{1,-}\}.
\end{equation}
Substituting this into Eq.~(\ref{eq:Bnd_app}) yields
\begin{equation} \braket{\delta\hat{q}^2} \geq \min_{\bs{a}\in\mathcal{A}}\frac{\norm{\bs{a}}_1^2}{4t^2}=\frac{\max\{\|\bs \alpha\|_{1,+}^2, \|\bs \alpha\|_{1,-}^2\}}{P^2t^2}, \end{equation}
which exactly matches the bound of Ref.~\cite{bringewatt_optimal_2024}.

\section{Hamiltonian Reshaping for Bosonic Phase Sensing}\label{app:bos_reshaping}

In this appendix, we argue that bosonic reshaping can be used for distributed bosonic phase estimation~\cite{proctor_multiparameter_2018, ge_distributed_2018, Guo2020:DQSphase,Liu2021:DQSphase, Kim2024:dqsPhase, bringewatt_optimal_2024} in interacting bosonic sensor networks by deriving an average-case reshaping error bound, quantified in the trace norm, for the effective evolution. This bound plays the role for bosonic systems that the QDRIFT-like Hamiltonian-simulation bound plays for finite-dimensional systems~\cite{campbell_random_2019}, and it captures the tradeoff between sampling complexity and the fidelity of the reshaped dynamics.

Consider an interacting bosonic sensor network of $m$ modes with number operators $\hat{\bm{n}}=(\hat n_1, \hat{n}_2, \dots, \hat{n}_m)$. Suppose there are a set of phases $\bm{\theta}=(\theta_1, \theta_2, \dots, \theta_m)$ encoded into the modes via
\begin{equation}\label{eq:og-Hboson}
 \hat{H}_0= \bm{\theta}\cdot\hat{\bm{n}} + \hat{H}_{\rm int}.
\end{equation}
As usual, our goal is to estimate a linear combination of the phases (or more appropriately, frequency shifts) $q=\bm{\alpha}\cdot\bm{\theta}$ to the best of our abilities, in spite of the unwanted interactions $\hat{H}_{\rm int}$. It turns out that we can perform a $\mathcal{U}(1)^m$-twirl in order to erase most of $\hat{H}_{\rm int}$, with the exception of a residual piece that is diagonal in the local Fock basis---thus reducing the problem to Diagonal Learning. Then, we can estimate the distributed phases using the Diagonal Learning protocol described in the main text,
which can achieve the limit predicted by the quantum Fisher information (QFI). This works because the error in QDRIFT simulating the interaction-free sensor state decreases polynomially with the number of random gates applied and can be made arbitrarily small, as we discuss below. 

As an introduction, consider a single-bosonic mode and the group $\mathcal{U}(1)$ generated by the bosonic number operator $\hat{n}=\hat a^\dagger \hat a$. 
A unitary representation of the group is $\hat{u}(\varphi) = e^{i\varphi \hat n}$ with $\varphi \in [0,2\pi)$. We apply these phases randomly according to the Haar measure, which is the flat distribution $\dd\mu(\varphi)=\dd\varphi/2\pi$. For a single bosonic mode, we write out the effective (symmetrized) Hamiltonian formally as
\begin{equation}
 \hat{H}_{\rm eff} = \int_0^{2\pi} \dd{\mu(\varphi)} \; \hat{u}(\varphi) \hat{H}_0 \hat{u}^\dagger(\varphi) \eqqcolon \mathbb{E}_{\varphi}[\hat{H}],
\end{equation}
This $\mathcal{U}(1)$-twirl projects $\hat{H}_0$ onto its number-conserving part, i.e., $[\hat{H}_{\rm int},\hat{n}]=0$. Likewise, for $m$ modes, the relevant group is $\mathcal{U}(1)^m$, with elements $\hat{u}(\bm\varphi)=\exp(i\bm{\varphi}\cdot \hat{\bm{n}})$. Averaging with the $m$-product Haar measure removes all terms that change local occupation numbers, such that $\,\forall\, i$ we have $[\hat{H}_{\rm eff},\hat{n}_i]=0$.

Consider the Hamiltonian from Eq.~\eqref{eq:og-Hboson}. We write the effective Hamiltonian as $\hat{H}_{\rm eff}= \bm{\theta}\cdot\hat{\bm{n}} + \hat{H}_{\rm Kerr}$, where $\hat{H}_{\rm Kerr}$ is diagonal in the local Fock basis. In particular, the effective interaction $\hat{H}_{\rm Kerr}$ must commute with \textit{any} element $\hat{u}(\bm{\varphi})=\exp(i\bm{\varphi}\cdot\hat{\bm{n}})$ and, therefore, consists of functions of $\{\hat{n}_i\}$, like Kerr terms $\hat{n}_i\hat{n}_j$. Any hopping terms $\hat a_i \hat a_j^\dagger$, squeezing terms $\hat a_i^\dagger \hat a_j^\dagger$, or linear displacements $\hat a_i$ vanish, since only terms commuting with local particle numbers $\hat n_i$ survive. In other words, bilinear interactions (apart from those linear in $\hat{n}_i$) vanish via $\mathcal{U}(1)^m$-twirling. If $\hat{H}_{\rm Kerr}$ is absent altogether, then we may directly apply distributed bosonic phase-estimation techniques~\cite{proctor_multiparameter_2018, ge_distributed_2018, bringewatt_optimal_2024} to estimate the linear combination $q=\bm{\alpha}\cdot\bm{\theta}$. Otherwise, the Diagonal Learning protocol developed here can be applied for bosonic phase sensing (see also Appendix~\ref{app:boson-diaglearning}).

We have shown that $\mathcal{U}(1)^m$-twirling effectively reduces our interacting Hamiltonian down to the sensing Hamiltonian $\hat{H}(\bm\theta)=\bm{\theta}\cdot\bm{\hat{n}}$, up to a Kerr-type term. However, there awaits an argument for the reshaping error of the effective evolution. We can apply a generalization of standard techniques for QDRIFT for finite-dimensional systems~\cite{campbell_random_2019,chen_concentration_2021} to derive a rigorous (average-case) error bound. In what follows, we provide a proof for the error bounds adapted to bosonic systems under the assumption that the Taylor series of sensor dynamics in time converges in trace norm.

Let $\hat \rho$ be a bosonic sensor state, $\hat{V}_{\bm{\varphi}}({\Delta t})= \hat{u}(\bm{\varphi}) \exp(-i\hat{H}_0{\Delta t}) \hat{u}^\dagger(\bm{\varphi})$ denote the single-step conjugated unitary, with $\hat{u}(\bm{\varphi})=\exp(i\bm{\varphi}\cdot\bm{\hat{n}})$ a unitary representation of the group $\mathcal{G}=\mathcal{U}(1)^m$, and $\Phi_{\Delta t}(\hat \rho)=\int\dd{\mu(\bm\varphi)}\hat{V}_{\bm{\varphi}}({\Delta t}) \hat \rho \hat{V}_{\bm{\varphi}}^\dagger({\Delta t})$. The target effective dynamics reads $\hat{U}_{\rm eff}({\Delta t})=\exp(-i\hat{H}_{\rm eff}{\Delta t})$. We denote off-diagonal contributions to the sensor Hamiltonian as $\delta\hat{H}= \hat{H}_0 - \hat{H}_{\rm eff}$, where $\hat{H}_{\rm eff}=\mathbb{E}_{\bm{\varphi}}[\hat{H}]$.

We find that 
\begin{equation}\label{eq:boson-bound1}
 \norm{\Phi_{\Delta t}(\hat\rho)-\hat{U}_{\rm eff}({\Delta t})\hat\rho \hat{U}_{\rm eff}^\dagger({\Delta t})}_1 \leq 2{\Delta t}^2 \norm{\mathbb{E}_{\bm{\varphi}}[\delta\hat{H}^2]\hat\rho}_1 + O(\Delta t^3),
\end{equation}
where $\delta\hat{H}= \hat{H}_0 - \hat{H}_{\rm eff}$ is the $\mathcal{U}(1)^m$-symmetry-breaking part of $\hat{H}_0$ and $\hat{H}_{\rm eff}=\mathbb{E}_{\bm{\varphi}}[\hat{H}]$. We analyze the trace norm $\norm{ \hat A}_1=\Tr[|\hat A|]$ with $|\hat A|=\sqrt{\hat A^\dagger \hat A}$ since $\hat{H}_0$ is generally unbounded (i.e., the spectral norm is infinite), thus introducing small subtleties in the analysis. It follows that
\begin{equation}\label{eq:boson-bound}
 \norm{\Phi_t(\hat\rho)-\hat{U}_{\rm eff}(t) \hat \rho \hat{U}_{\rm eff}^\dagger(t)}_1 \leq \frac{2t^2}{L} \norm{\mathbb{E}_{\bm{\varphi}}[\delta\hat{H}^2] \hat\rho}_1 + O(t^3/L^2),
\end{equation}
where $t=L{\Delta t}$ and $\Phi_t(\hat\rho)=\Phi_{\Delta t}^L(\hat\rho)$ represents $L$ compositions of $\Phi_{\Delta t}$. This bound is extremely similar to Eq.~\eqref{eq:bias_error} but now with dependence on the quantum sensor state $\hat\rho$. Again, such state dependence is expected since the bosonic Hilbert-space dimension is unbounded. If we constrain the total photon number of the sensor state to be $n _{\text{ph}}$ and the highest-order terms in $\delta\hat{H}$ (in bosonic raising and lowering operators) to be $R$, then we expect the simulation error to be bounded as 
\begin{align}
\norm{\Phi_t(\hat \rho)-\hat{U}_{\rm eff}(t)\hat \rho \hat{U}_{\rm eff}^\dagger(t)}_1 = 
O \left ( \frac{ t^2}{L} n _{\text{ph}}^{2R} \right )
.
\end{align}
\begin{proof}
 We first prove the small step Eq.~\eqref{eq:boson-bound1} and then Eq.~\eqref{eq:boson-bound} by extension. Our proof makes use of a Taylor expansion of the evolution and thus implicitly assumes that the corresponding series converges in the trace norm. This assumption is reasonable provided that $\hat \rho$ has sufficiently well-behaved moments of $\hat{H}_{0}$.

 For convenience, define the map $\mathscr{U}_{\rm eff}(\cdot)=\hat{U}_{\rm eff}({\Delta t})\cdot \hat{U}_{\rm eff}^\dagger({\Delta t})$. Expand $\hat{U}_{\rm eff}({\Delta t})$ and $\hat{V}_{\bm{\varphi}}({\Delta t})$ to $O(\Delta t^2)$:
 \begin{align}
 \mathscr{U}_{\rm eff}(\hat \rho)&=\hat \rho - i{\Delta t}\comm{\hat{H}_{\rm eff}}{\hat \rho}+{\Delta t}^2\left(\hat{H}_{\rm eff} \hat \rho \hat{H}_{\rm eff}-\frac{1}{2}\acomm{\hat{H}_{\rm eff}^2}{\hat \rho}\right) + O(\Delta t^3), \\
 \hat{V}_{\bm{\varphi}}({\Delta t})\hat \rho \hat{V}_{\bm{\varphi}}^\dagger({\Delta t}) &= \hat \rho - i{\Delta t}\comm{\hat{H}_{\bm{\varphi}}}{\hat \rho}+{\Delta t}^2\left(\hat{H}_{\bm{\varphi}} \hat \rho \hat{H}_{\bm{\varphi}}-\frac{1}{2}\acomm{\hat{H}_{\bm{\varphi}}^2}{\hat \rho}\right) + O(\Delta t^3),
 \end{align}
 with shorthand $\hat{H}_{\bm{\varphi}}=\hat{u}(\bm\varphi)\hat{H}_0\hat{u}^\dagger(\bm\varphi)$. Averaging over the Haar measure $\dd{\mu(\bm\varphi)}$:
 \begin{equation}\label{eq:diffrho1}
 \Phi_{\Delta t}(\hat\rho)-\mathscr{U}_{\rm eff}(\hat\rho)={\Delta t}^2\left(\mathbb{E}_{\bm{\varphi}}\left[\hat{H}\hat\rho \hat{H}-\frac{1}{2}\acomm{\hat{H}^2}{\hat\rho}\right]-\hat{H}_{\rm eff} \rho \hat{H}_{\rm eff}+\frac{1}{2}\acomm{\hat{H}_{\rm eff}^2}{\hat\rho}\right) + O(\Delta t^3).
 \end{equation}
 Generically, we may write $\hat{H}=\hat{H}_{\rm eff} + \delta\hat{H}$, with $\delta\hat{H}$ representing the symmetry-breaking piece of $\hat{H}_0$, such that $\mathbb{E}_{\bm{\varphi}}[\delta\hat{H}]=0$. Then Eq.~\eqref{eq:diffrho1} simplifies to
 \begin{equation}\label{eq:diffrho2}
 \Phi_{\Delta t}(\hat\rho)-\mathscr{U}_{\rm eff}(\hat\rho) = {\Delta t}^2 \mathbb{E}_{\bm{\varphi}}\left[\delta\hat{H}\hat\rho \delta\hat{H}-\frac{1}{2}\acomm{\delta\hat{H}^2}{\hat\rho}\right] + O(\Delta t^3).
 \end{equation}
 Invoking \textit{(i)} the triangle inequality, \textit{(ii)} $\norm{\hat{X}}_1 = \norm{\hat{X}^\dagger}_1$ for any $\hat{X}$, and \textit{(iii)} Hermiticity of $\delta\hat{H}$ and $\hat\rho$, we find
 \begin{equation}\label{eq:diff-norm1}
 \norm{\Phi_{\Delta t}(\hat\rho)-\mathscr{U}_{\rm eff}(\hat\rho)}_1 \leq {\Delta t}^2 \left(\norm{\mathbb{E}_{\bm{\varphi}}[\delta\hat{H}\hat\rho \delta\hat{H}]}_1 + \norm{\mathbb{E}_{\bm{\varphi}}[\delta\hat{H}^2\hat\rho]}_1\right) + O(\Delta t^3).
 \end{equation}
 Given $\norm{\hat{X} \hat\rho \hat{X}}_1 \leq \norm{\hat{X}^2 \hat\rho}_1$ for $\hat{X}$ Hermitian, one may show that $\norm{\mathbb{E}_{\bm{\varphi}}[\delta\hat{H}\hat\rho \delta\hat{H}]}_1 \leq \norm{\mathbb{E}_{\bm{\varphi}}[\delta\hat{H}^2]\hat\rho}_1$. We substitute this inequality into the right-hand-side of Eq.~\eqref{eq:diff-norm1} and compute
 \begin{equation}\label{eq:diff-norm2}
 \norm{\Phi_{\Delta t}(\hat\rho)-\mathscr{U}_{\rm eff}(\hat\rho)}_1 \leq 2{\Delta t}^2 \norm{\mathbb{E}_{\bm{\varphi}}[\delta\hat{H}^2]\hat\rho}_1 + O(\Delta t^3).
 \end{equation}
 
This is for a single step. We need a relation for the total time $t=L\Delta t$. Construct the map $\Phi_t=\Phi_{\Delta t}^L$, which we compare to $\mathscr{U}_{\rm eff}^L(\cdot)=\hat{U}_{\rm eff}(t)\cdot \hat{U}_{\rm eff}^\dagger(t)$. Apply the telescoping identity,
 \begin{equation}
     \Phi_{t} - \mathscr{U}_{\rm eff}^L = \sum_{k=0}^{L-1} \Phi_{\Delta t}^{\,L-1-k}\,(\Phi_{\Delta t}-\mathscr{U}_{\rm eff})\,\mathscr{U}_{\rm eff}^{\,k},
 \end{equation}
to the initial state $\hat\rho$ and take the trace norm. Then, using the triangle inequality and contractivity of completely positive trace-preserving maps:
\begin{align}
\norm{\Phi_{t}(\hat\rho)-\mathscr{U}_{\rm eff}^L(\hat\rho)}_1
&= \norm{ \sum_{k=0}^{L-1} \Phi_{\Delta t}^{\,L-1-k}(\Phi_{\Delta t}-\mathscr{U}_{\rm eff})\mathscr{U}_{\rm eff}^{\,k}(\hat\rho)}_1 \\
&\leq \sum_{k=0}^{L-1}  \norm{\Phi_{\Delta t}^{\,L-1-k}(\Phi_{\Delta t}-\mathscr{U}_{\rm eff})\mathscr{U}_{\rm eff}^{\,k}(\hat\rho)}_1\\ 
&\leq \sum_{k=0}^{L-1} \norm{(\Phi_{\Delta t}-\mathscr{U}_{\rm eff})\mathscr{U}_{\rm eff}^{\,k}(\hat\rho)}_1 \\
&=  \sum_{k=0}^{L-1} \norm{\Phi_{\Delta t}\left(\mathscr{U}_{\rm eff}^{\,k}(\hat\rho)\right)-\mathscr{U}_{\rm eff}\left(\mathscr{U}_{\rm eff}^{\,k}(\hat\rho)\right)}_1.
\end{align}
Invoking Eq.~\eqref{eq:diff-norm2} to each summand with input state $\mathscr{U}_{\rm eff}^{\,k}(\hat\rho)$:
\begin{equation}
    \norm{\Phi_{t}(\hat\rho)-\mathscr{U}_{\rm eff}^L(\hat\rho)}_1
 \leq \sum_{k=0}^{L-1} 2{\Delta t}^2 \norm{\mathbb{E}_{\bm{\varphi}}[\delta\hat{H}^2]\mathscr{U}_{\rm eff}^{\,k}(\hat\rho)}_1 + O(L\Delta t^3). 
\end{equation}
Now observe that any operator $\mathbb{E}_{\bm{\varphi}}[\hat{O}]$ is diagonal in the local Fock basis and, thus, commutes with $\hat{U}_{\rm eff}$ since $\hat{H}_{\rm eff}$ is likewise diagonal in the local Fock basis. Therefore, $\mathbb{E}_{\bm{\varphi}}[\delta\hat{H}^2]\mathscr{U}_{\rm eff}^{\,k}(\hat\rho)=\mathscr{U}_{\rm eff}^{\,k}(\mathbb{E}_{\bm{\varphi}}[\delta\hat{H}^2]\hat\rho)$, which implies
\begin{align}
    \norm{\Phi_{t}(\hat\rho)-\mathscr{U}_{\rm eff}^L(\hat\rho)}_1 &\leq \sum_{k=0}^{L-1} 2{\Delta t}^2 \norm{\mathbb{E}_{\bm{\varphi}}[\delta\hat{H}^2]\mathscr{U}_{\rm eff}^{\,k}(\hat\rho)}_1 + O(L\Delta t^3) \\
    &= \sum_{k=0}^{L-1} 2{\Delta t}^2 \norm{\mathscr{U}_{\rm eff}^{\,k}(\mathbb{E}_{\bm{\varphi}}[\delta\hat{H}^2]\hat\rho)}_1 + O(L\Delta t^3) \\
    &= \sum_{k=0}^{L-1} 2{\Delta t}^2 \norm{\mathbb{E}_{\bm{\varphi}}[\delta\hat{H}^2]\hat\rho}_1 + O(L\Delta t^3) \\
    &= 2 L {\Delta t}^2 \norm{\mathbb{E}_{\bm{\varphi}}[\delta\hat{H}^2]\hat\rho}_1 + O(L\Delta t^3).
\end{align}
The third line derives from unitary invariance of the trace norm. Since $t=L\Delta t$, the result [Eq.~\eqref{eq:boson-bound}] follows.
\end{proof}

\section{Alternative Protocols for Interacting Sensing}\label{app:alternative-protocols}
Here, we study two alternative protocols for the Interacting Sensing problem, with various restrictions on the allowed resources in each case.

For our first alternative protocol, we use a product initial state, apply only single-qubit control during evolution, and perform single-qubit measurements. For estimating the parameter $q = \sum_i \alpha_i \theta_i$, we measure each $\theta_i$ individually and then compute the corresponding weighted sum. To measure $\theta_i$ on the $i$-th qubit, we reshape the Hamiltonian to an effective form $\theta_i \hat{Z}_i$ by applying random Pauli operators to all other qubits. 
As a result, the effective Hamiltonian becomes
\begin{equation}
 \hat{H}_\text{eff} = \theta_i \hat{Z_i}.
\end{equation}
We then prepare the $i$-th qubit in the $\ket{+}$ state and let it evolve under this effective Hamiltonian for a time $t/n$. This procedure allows us to estimate $\theta_i$ with variance ${n^2}/{4\nu t^2}$, where $\nu$ is the number of repetitions. Repeating this procedure for each $i$, we estimate $q$ with total variance
\begin{equation}
\braket{\delta\hat{q}^2} = \frac{Q_1\| \bs{a}\|^2_1}{4\nu t^2}, \quad Q_1 = n^2 \frac{ \| \bs\alpha\|^2_2}{\|\bs a\|_1^2},
\end{equation}
where $Q_1$ denotes the relative variance compared to the optimal protocol. Note this is not an optimal protocol for $n>1$, since $\norm{\bs{a}}_1 \leq \sqrt{N}\norm{\bs{a}}_2 = \norm{\bs{\alpha}}_2$, with the latter equality following from Plancherel's theorem. Thus, $Q_1 \geq n^2$, which provides an upper bound on the best achievable variance for $\hat{q}_{\text{est}}$ when restricting to single-qubit control and measurement.

In the second alternative protocol, we again begin with a product state and single-qubit control, but conclude with an entangled measurement. 
 As in the optimal protocol, by using only single-qubit control, we can reshape the Hamiltonian such that its generators are in the group $\mathcal{S} = \{ \hat{I},\hat{Z}\}^{\otimes n}$. The resulting effective Hamiltonian is
\begin{equation}
 \hat{H}_\text{eff} = \sum_{i=1}^n \theta_i \hat{Z}_i + \sum_{j\,|\,\hat{s}_j \in \mathcal{S}/\{\hat{Z}_i\}} \gamma_j \hat{s}_j. \label{eq:H_prod}
\end{equation}
Consider the initial state $\ket{+}^{\otimes n}$. Assuming no additional control or ancilla qubits, the system evolves under $\hat{H}_\text{eff}$ for time $t$, leading to the state $\ket{\psi(t)} = \hat{U}_{\rm eff}(t) \ket{+}^{\otimes n}$, where $\hat{U}_{\rm eff}(t) = e^{-i \hat{H}_\text{eff} t}$. This state encodes full information about any linear function of $\bs{\theta}$, as measuring in the $X$ basis allows one to learn all the parameters of the Hamiltonian in \cref{eq:H_prod}. 
After $\nu$ repetitions, the quantum Fisher information matrix in the limit $q \rightarrow 0$ may be computed to be $\mathcal{F}_{ij} = 4\delta_{ij}$, so that the quantum Cram{\'e}r-Rao bound is given by
\begin{equation}
\braket{\delta\hat{q}^2} \geq \bs\alpha^T \mathcal{F}^{-1} \bs\alpha = 
 \frac{Q_2\norm{\bs a}_1^2}{4\nu t^2}, \quad Q_2 = \frac{\| \bs\alpha\|^2_2}{\|\bs a\|_1^2},
\end{equation}
and we generically have $Q_2 \geq 1$.
To saturate this bound, an optimal measurement must be performed. Such a measurement is given by the projective measurement onto the state $\ket{\varphi}$, defined by
\begin{equation}
 \begin{cases}
 \hat{\mathcal{M}}_0 = \hat{I} - \dyad{\varphi} \\ 
 \hat{\mathcal{M}}_1 = \dyad{\varphi}
 \end{cases}, \quad 
 \ket{\varphi} = \frac{1}{\| \bs \alpha\|_2} \sum_{i=1}^n \alpha_i \hat{Z}_i \ket{+}^{\otimes n}.
\end{equation}
Computing the corresponding variance after $\nu$ independent measurements for $\hat{q}_\text{est}$ yields
\begin{equation}
\braket{\delta\hat{q}^2} = \frac{\| \bs \alpha\|_2^2}{4\nu t^2},
\end{equation}
saturating the above Cramér–Rao bound for our product initial state. Thus, allowing entangled measurements at the end of each experiment yields a precision improvement compared to protocols limited to single-qubit control and measurement.

\end{document}